\documentclass[envcountsect]{svmult}

\usepackage{makeidx}         
\usepackage{graphicx}        
\usepackage{multicol}        
\usepackage[bottom]{footmisc}
\usepackage{latexsym}
\usepackage{amsmath}
\usepackage{url,color}
\usepackage{bm}
\usepackage{caption}
\usepackage{amssymb}

\newcommand{\vP}{{\bf P}}
\newcommand{\vR}{{\bf R}}
\newcommand{\bzero}{{\mathbf 0}}

\date{\today}

\newcommand{\be}{\begin{eqnarray}}
\newcommand{\ee}{\end{eqnarray}}

\newcommand{\R}{\mathbb{R}}
\newcommand{\half}{\frac{1}{2}}
\newcommand{\ep}{\varepsilon}

\newcommand{\om}{\Omega}
\newcommand{\Div}{{\rm div}\,}

\newcommand{\1}{{\bf 1}}
\newcommand{\tr}{{\rm tr}\,}
\newcommand{\cof}{{\rm cof}\,}

\newcommand{\diag}{{\rm diag}\,}
\newcommand{\n}{{\bf n}}
\newcommand{\N}{{\bf N}}
\newcommand{\x}{{\bf x}}
\newcommand{\z}{{\bf z}}
\newcommand{\p}{{\bf p}}
\newcommand{\q}{{\bf q}}
\newcommand{\m}{{\bf m}}
\newcommand{\M}{{\bf M}}
\newcommand{\B}{{\bf B}}

\newcommand{\vE}{{\bf E}}
\newcommand{\Q}{{\bf Q}}

\newcommand{\az}{{\bf a}}
\newcommand{\bz}{{\bf b}}

\newcommand{\e}{{\bf e}}
\newcommand{\uu}{{\bf u}}
\newcommand{\y}{{\bf y}}
\newcommand{\A}{{\bf A}}
\newcommand{\vD}{{\bf D}}

\newcommand{\nnu}{{\bm \nu}}
\newtheorem{thm}{Theorem}
\newtheorem{cor}[thm]{Corollary}
\newtheorem{lem}[thm]{Lemma}
\newtheorem{prop}[thm]{Proposition}

\newtheorem{rem}{Remark}

\graphicspath{%
    {converted_graphics/}
    {C:/Users/John/Documents/Papers/}
    {/}
    {C:/Users/John/Pictures/}
}
\makeindex             

\begin{document}

\title*{Liquid crystals and their defects}
\author{J. M. Ball}
\institute{Oxford Centre for Nonlinear PDE,\\ Mathematical Institute,   University of Oxford,\\
Andrew Wiles Building,
Radcliffe Observatory Quarter,\\
Woodstock Road,
Oxford,
OX2 6GG,
 U.K.\\
\texttt{ball@maths.ox.ac.uk}
}
%
%
\maketitle
 \begin{abstract}
These lectures describe some classical models of liquid crystals, the relations between them, and the different ways in which these models describe defects. 
\end{abstract}

\section{Introduction}
\label{intro}
\setcounter{equation}{0}
This course of lectures discusses classical models of liquid crystals, and the different ways in which defects are described according to the different models. By a defect we mean  a point, curve or surface, in the neighbourhood of which the order parameter describing the orientation of the liquid crystal molecules varies very rapidly.  Defects can be observed optically, for example using polarized light, but it is difficult to obtain definitive information about their small-scale structure via microscopy. Depending on the theory used, a defect may or may not be represented by a mathematical singularity in the order parameter field. One of the themes running through the lectures is the importance of a proper function space setting for the  description of  defects. The same energy functional may predict different behaviour according to the function space used, as this space may allow the description of one kind of defect but not another. 

This is part of a more general issue concerning continuum models of physics, which are not complete without specification of a function space describing the allowed singularities. Said differently, {\it the function space is part of the model}. 
Of course these questions are closely related to the different possible levels of description for materials (atomic, molecular, continuum ...) and how these can be reconciled. The more detailed the description the larger the dimension of the corresponding order parameter. Textbook derivations of models of continuum physics do not usually pay much attention to function spaces, explicitly or implicitly assuming that continuum variables are smooth, and it is only when analysts start trying to prove existence, uniqueness and regularity theorems that function spaces start to proliferate. This  obviously unsatisfactory state of affairs  glosses over the deep question of where the function space comes from, which is intimately connected to the relation between models having more or less  detailed levels of description.

These issues are perhaps better understood for solid mechanics than for liquid crystals (for more discussion see \cite{j70}),   a  particularly strong analogy being with function spaces used to describe different kinds of fracture in solids.

A further theme is the treatment of equality and inequality constraints in models of continuum physics, and how they are preserved by solutions. Again there is an interesting comparison to be made between liquid crystals and solid mechanics, where there are similar open problems concerning the preservation of the eigenvalue constraints on the $\Q-$tensor for the Landau - de Gennes theory, and positivity of the Jacobian (related to non-interpenetration of matter) in nonlinear elasticity.

The study of liquid crystals is an interdisciplinary subject in which aspects of chemistry, physics, engineering, mathematics and computer simulation are all necessary for a full understanding. The interaction with mathematics, in particular algebra, geometry, topology and partial differential equations, continues to be a source of deep and interesting problems, and I hope these notes will help to attract researchers to some of these.

For general introductions to the mathematics of liquid crystals the reader is referred to the texts of Stewart \cite{stewart04} and Virga \cite{virgabook}, and for a compendium of classic papers in the subject to Sluckin, Dunmur \& Stegemeyer  \cite{sluckin2004crystals}. For a comprehensive review of liquid crystal defects see Kl\'eman \cite{Kleman1989}.

\section{What are liquid crystals?}
\label{lcs}
Liquid crystal phases are states of matter  intermediate between crystalline solids and isotropic fluids. The  interaction of these phases with electromagnetic fields has led to a multi-billion dollar industry centred around the ubiquitous liquid crystal displays (LCDs)  found in billions of PCs and laptops, televisions and watches. The characteristic properties of liquid crystal phases originate from the shape and other properties of their constituent molecules and the interactions between them. 

In the most common {\it thermotropic} liquid crystals the liquid crystalline phases typically exist in a temperature range above which  the material behaves like an isotropic fluid, and below which it behaves like a solid. The liquid crystal phases are characterized by orientational order of their constituent molecules,  with in some cases a limited amount of positional order, and they form a special kind of {\it nonlinear fluid}. Commercial liquid crystals usually comprise a mixture of different kinds of molecules to optimize performance. In {\it lyotropic} liquid crystals, which we do not consider further in this course,  the liquid crystal phases depend both on temperature and the concentration of the liquid crystal molecules in a solvent, such as water.

Typical thermotropic liquid crystals, such as MBBA  and 5CB, consist of  molecules having lengths of the order of 2-3nm. It is instructive to look at 3D space-filling models of such molecules, in which atoms are represented by spheres with radius proportional to the radius of the atom, see Fig. \ref{MBBA5CBa}. 
\begin{figure}[tbp] 
  \centering
  \includegraphics[width=4.09in,height=0.8in,keepaspectratio]{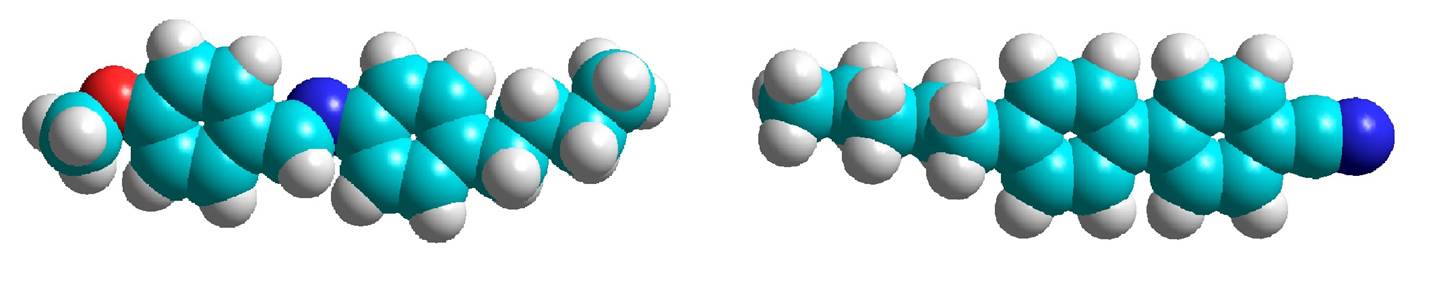}
  \caption{Space-filling models of liquid crystal molecules (courtesy C. Zannoni): (a) MBBA (N-(4-methoxybenzylidene)-4-butylaniline), (b) 5CB (4-Cyano-4$'$-pentylbiphenyl)}
  \label{MBBA5CBa}
\end{figure}
Such molecules have approximate rod-like shapes, and are often  idealized as ellipsoids of revolution.  

There are three main liquid crystal phases, {\it nematics}, {\it cholesterics} and {\it smectics}. In the nematic phase the molecules have orientational order but no positional order, so that the mean orientation of the long axis of the molecules at the point $\x$ and time $t$  can be represented by a unit vector $\n=\n(\x,t)$ called the {\it director} (see Fig. \ref{nemsmectic}(a)).  In the cholesteric (or chiral nematic) phase the molecules form a helical structure with an axis perpendicular to the local director (see Fig. \ref{cholesteric}).
 The smectic phases have orientational and  some positional order. In the smectic A phase the molecules arrange themselves in layers of the order of a molecular length in thickness, with the director $\n$ parallel to the layer normal $\m$ (see Fig. \ref{nemsmectic}(b)). The molecules may move between layers. In the smectic C phase (see Fig. \ref{nemsmectic}(c)) the director makes a fixed angle with the layer normal. The molecules are in thermal motion, so that Fig. \ref{nemsmectic} is a schematic representation at a fixed time. \begin{figure}[tbp] 
  \centering
  \includegraphics[width=4.09in,height=1.67in,keepaspectratio]{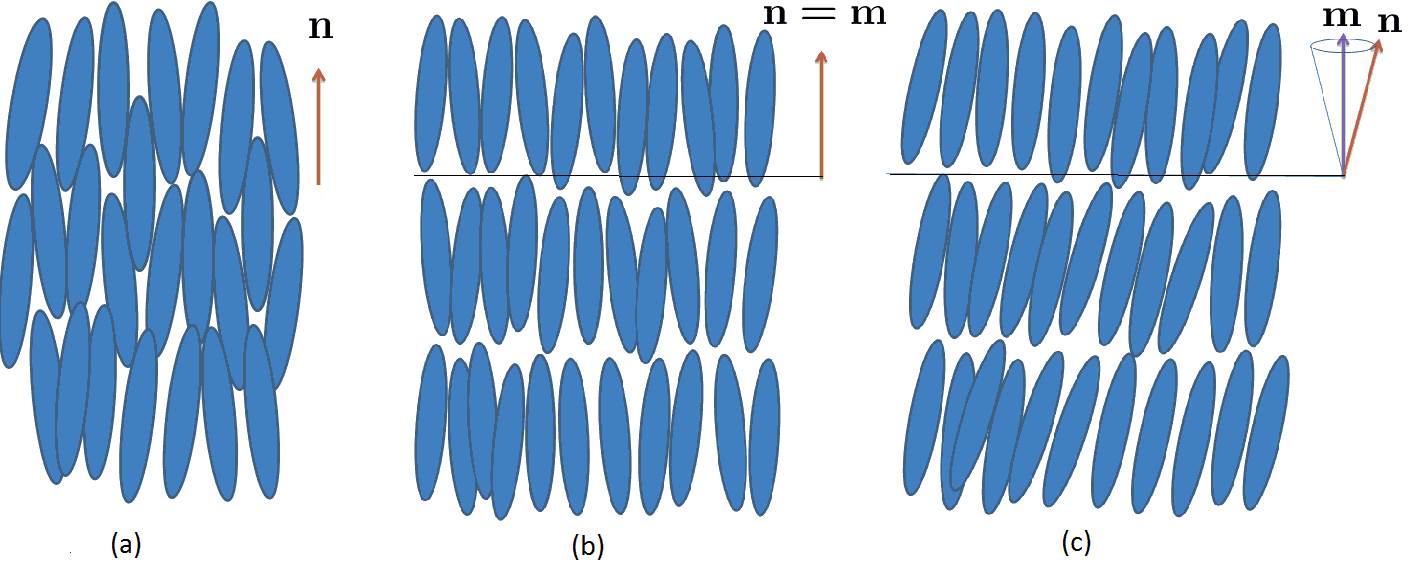}
  \caption{Nematic and smectic liquid crystal phases: (a) nematic phase, with director $\bf n, |\n|=1,$ giving the mean orientation of molecules, (b) smectic A phase, with $\bf n$ parallel to the layer normal $\bf m$, (c) smectic C phase, in which $\bf n$ makes a fixed angle with $\bf m$.}
  \label{nemsmectic}
\end{figure}
\begin{figure}[tbp] 
  \centering
  \includegraphics[width=2.00in,height=2in,keepaspectratio]{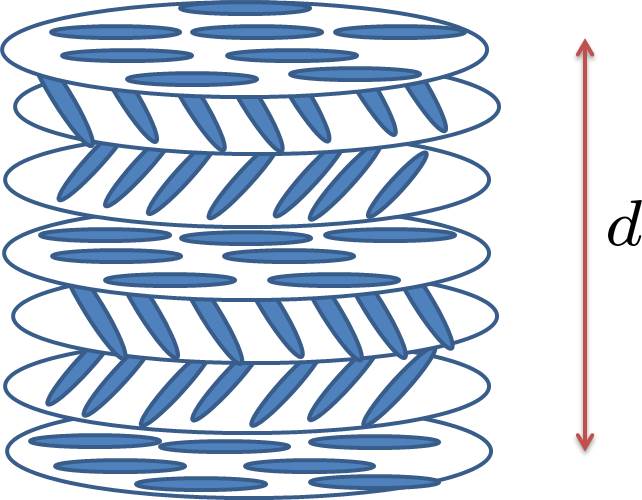}
  \caption{Cholesteric phase: the pitch $d$ corresponds to the distance, typically of the order of microns, over which the mean orientation of molecules rotates by $2\pi$.}
\label{cholesteric}
\end{figure}
There are other possible phases such as smectic B, which is similar to smectic A but with hexagonal ordering in the layers.

The nematic phase typically arises on cooling through a critical temperature as a phase transition from a higher temperature isotropic phase, in which the molecules have no long-range orientational or positional order, as illustrated in Fig. \ref{isotropicpic}. 
\begin{figure}[tbp] 
  \centering
  \includegraphics[width=2.09in,height=2in,keepaspectratio]{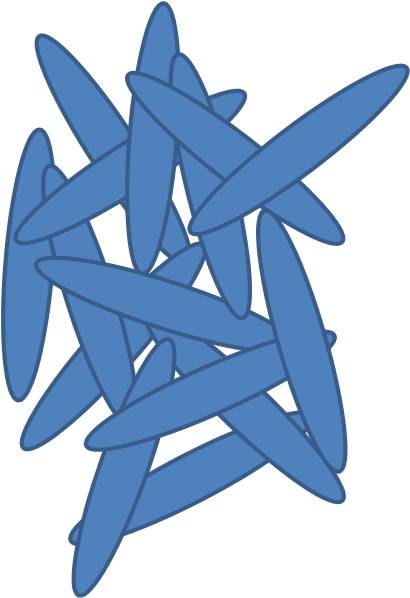}
  \caption{Isotropic phase with no orientational or positional order.}
  \label{isotropicpic}
\end{figure}
Thus for temperatures $\theta>\theta_c$ the material is an isotropic fluid, while for $\theta_m<\theta<\theta_c$ the material is in the nematic phase. For $\theta<\theta_m$ the material may be in another liquid crystal (e.g. smectic) or solid phase. For MBBA we have $\theta_m\sim 17^\circ C, \theta_c\sim 45^\circ C$. For videos of the isotropic to nematic phase transition see the website  
\url {\https://www.doitpoms.ac.uk/}.

Most liquid crystal displays are of {\it twisted nematic} type. A single pixel consists of nematic liquid crystal   confined between two parallel glass plates, at $x_3=0, \delta$ say, treated so that the director lies parallel to $\e_1$ on the plate $x_3=0$ and parallel to $\e_2$ on the plate $x_3=\delta$, where $\e_i$ denotes the unit vector in the $x_i$-direction. Assuming that these boundary conditions are exactly satisfied, the Oseen-Frank theory discussed later in these notes predicts that in equilibrium and in the absence of an applied electric field the director undergoes a pure twist having the form 
\be 
\label{twist}
\n(\x)=\left(\cos \frac{\pi x_3}{2\delta}, \sin \frac{\pi x_3}{2\delta},0\right).
\ee 
Attached to the glass plates are polarizers aligned at right-angles to each other and parallel to the easy axis prescribed on each plate. Plane polarized light passes through the first polarizer and is then twisted by the liquid crystal so that it passes through the second one, so that the pixel is bright. But if an electric field normal to the plates is applied, and if the nematic has been chosen to have a positive dielectric anisotropy, the molecules align parallel to the field, the light is not twisted, and the pixel is dark.

\section{Models and order parameters}
\label{modelsorderparameters}
\subsection{Molecular dynamics}Liquid crystals can be modelled with various degrees of precision. At a very detailed level one can describe and simulate the interactions between the atoms in each liquid crystal molecule, and between these atoms and those of other molecules, but of course such a detailed description is intractable for the very large number of molecules in typical applications. 

Somewhat more tractable is to carry out Monte Carlo or molecular dynamics simulations using an empirical potential for the interaction between molecules. One commonly used such potential is the Gay-Berne potential \cite{gayberne1981}, which models the molecules as ellipsoids of revolution, the interaction potential between a pair of molecules being a generalization of the Lennard-Jones potential between pairs of atoms or molecules that depends on the orientations of the ellipsoids and the vector joining their centres of mass. This potential predicts the existence of isotropic, nematic, smectic A and smectic B  phases (see, for example, \cite{migueletal1996,luckhurstsimmonds1993,zannoni2001}). It has been used in 
\cite{riccietal2010} to study the twisted nematic cell in a simulation using about $10^6$ molecules, confirming the twist structure and giving information on switching between the bright and dark states. 

Generally, given faith in the effectiveness of the potential, atomistic and molecular dynamics simulations can probe regions, such as near surfaces and defects, which are inaccessible to current microscopy, providing useful input to appropriate continuum models.

\subsection {Order parameters}
\label{orderparameters}
Despite the interest of molecular dynamics models, they are clearly inadequate for predicting and understanding macroscopic configurations of liquid crystals, for which a continuum description is essential. Among the variables necessary for such a continuum description are {\it order parameters} that describe the nature and degree of order in the liquid crystal. We have already introduced one such order parameter, the director $\n=\n(\x,t)$, a unit vector describing the mean orientation of the molecules at the point $\x$ and time $t$. In fact the sign of $\n$ has no physical meaning because of the statistical head-to-tail symmetry of the molecules, so that $\pm\n$ are physically equivalent. Thus it is better to think of the director not as a vector field but as a line field, i.e. for each $\x,t$ to identify the mean orientation of molecules with the line through the origin parallel to $\n(\x,t)$.  Lines through the origin form the {\it real projective plane} $\R P^2$, elements of which can be identified with antipodal pairs of unit vectors $\pm\p$ or with matrices $\p\otimes \p, \p\in S^2$, where $(\p\otimes\p)_{ij}=p_ip_j$ and $S^2=\{\p\in\R^3:|\p|=1\}$ denotes the unit sphere.

 In this course we will consider only static configurations of liquid crystals, in which the fluid velocity is zero (although the discussion of order parameters that follows applies more generally). Thus the continuum variables will depend on $\x$ and not on $t$. We represent a typical  liquid crystal molecule by a bounded open region $M\subset \R^3$ (rod, ellipsoid, parallepiped ...) of approximately the same shape and symmetry. We place $M$ in a standard position with centroid at the origin, so that 
\be 
\label{centroid}
\int_M\y\,d\y=\bzero.
\ee
 Denoting by $M^{3\times 3}$ the space of real $3\times 3$ matrices with inner product $\A\cdot\B=\tr \A^T\B$ and corresponding norm $|\A|=(\A\cdot\A)^\half$, we define the isotropy groups
\begin{eqnarray*}
G_M&=&\{\vR\in O(3):\vR M=M\},\\
G_M^+&=&\{\vR\in SO(3):\vR M=M\},
\end{eqnarray*}
where $O(3)=\{\vR\in M^{3\times 3}:\vR^T\vR=\1\}$ is the set of orthogonal matrices and $SO(3)=\{\vR\in O(3):\det \vR=1\}$ is the set of rotations\footnote{Here we consider only the shape of $M$ as being important. More generally we could require the invariance of a vector $\uu=\uu(\x), \x\in M,$ of additional molecular variables (such as   mass or charge density), defining corresponding isotropy groups $\tilde G_M=\{\vR\in O(3):\vR M=M, \uu(\vR\x)=\uu(\x) \mbox{ for all }\x\in M\}, \tilde G_M^+=\{\vR\in SO(3):\vR M=M, \uu(\vR\x)=\uu(\x) \mbox{ for all }\x\in M\}$.}. Note that by \eqref{centroid} the centroid of $\vR M$ is zero for all $\vR\in O(3)$.

We say that the molecule represented by $M$ is {\it chiral} (as in cholesterics) if no reflection of $M$ is a rotation of $M$, that is $(\1-2\e\otimes\e) M\neq \vR M$ for any unit vector $\e$ and any $\vR\in  SO(3)$, which is easily seen to be equivalent to the condition that $G_M=G_M^+$. 

Note that  $\vR M=\tilde \vR M$ for $\vR,\tilde \vR\in SO(3)$ if and only if  $\tilde \vR^T \vR\in G_M^+$. Hence the orientation of a molecule can be represented (c.f. Mermin \cite{mermin1979}) by an element of the (left) {\it space of cosets} $SO(3)/G_M^+$ consisting of the distinct sets $\vR G_M^+$ where $\vR\in SO(3)$. For $M$ a circular cylindrical rod or ellipsoid of revolution with long axis parallel to the unit vector $\e_1$ in the $x_1-$direction we have\footnote{For example, in the case of the ellipsoid of revolution $M=\{\x=(x_1,x_2,x_3): \frac{x_1^2}{a^2}+\frac{x_2^2+x_3^2}{b^2}<1\}$, with semimajor axes $a>0, b>0, a\neq b$, if $\hat\vR M=M$ then $\hat\vR\partial M=\partial M$, and since $|\pm\hat\vR a\e_1|=a$ and the only points of $\partial M$ distant $a$ from $\bzero$ are $\pm a \e_1$ we have that $\hat\vR \e_1=\pm\e_1$. Conversely, if $\hat\vR\e_1=\pm\e_1$ then it is easily checked that $\hat\vR M=M$.} that $G_M^+=\{\hat\vR\in SO(3): \hat\vR\e_1=\pm\e_1\}$. Hence in these cases each coset $\vR G_M^+$ has the form
\begin{eqnarray*}
\vR G_M^+&=&\{\vR\hat\vR: \hat\vR\in SO(3), \hat\vR \e_1=\pm\e_1\}\\
&=&\{\tilde\vR:\tilde \vR\in SO(3):\tilde \vR\e_1=\pm\vR\e_1\}
\end{eqnarray*}
and thus can be identified with the rotations $\tilde \vR$ mapping the line joining $\pm \e_1$ to the line joining $\pm\p$, where $\p=\vR\e_1$. Hence the possible orientations of such a molecule can be identified with the elements $\p\otimes\p$ of $\R P^2$, as is intuitively clear.

Consider a liquid crystal filling a container $\om$, which we take to be a bounded open subset of $\R^3$ having sufficiently regular (e.g. Lipschitz) boundary. We suppose that the liquid crystal molecules are rod-like, so that as described above their orientations can be identified with elements of $\R P^2$. We adopt a coarse-graining procedure which up to now has not been justified rigorously.
Pick a point $\x\in \om$ and a small radius $\delta>0$. We suppose that $\delta$ is   sufficiently small so that  the ball $B(\x,\delta)$ can be identified with the material point $\x$ (i.e. $\delta$ is small on a  macroscopic length-scale), but large enough to contain enough molecules for a statistical description to be valid. To get an idea of the orders of magnitude, if $\delta=1\mu$m then for a typical liquid crystal $B(\x,\delta)$ will contain about $10^9$ molecules. Picking molecules at random from those $N=N(\x)$ molecules lying entirely within $B(\x,\delta)$ (see Fig. \ref{picking})
\begin{figure}[tbp] 
  \centering
  \includegraphics[width=4.17in,height=3.48in,keepaspectratio]{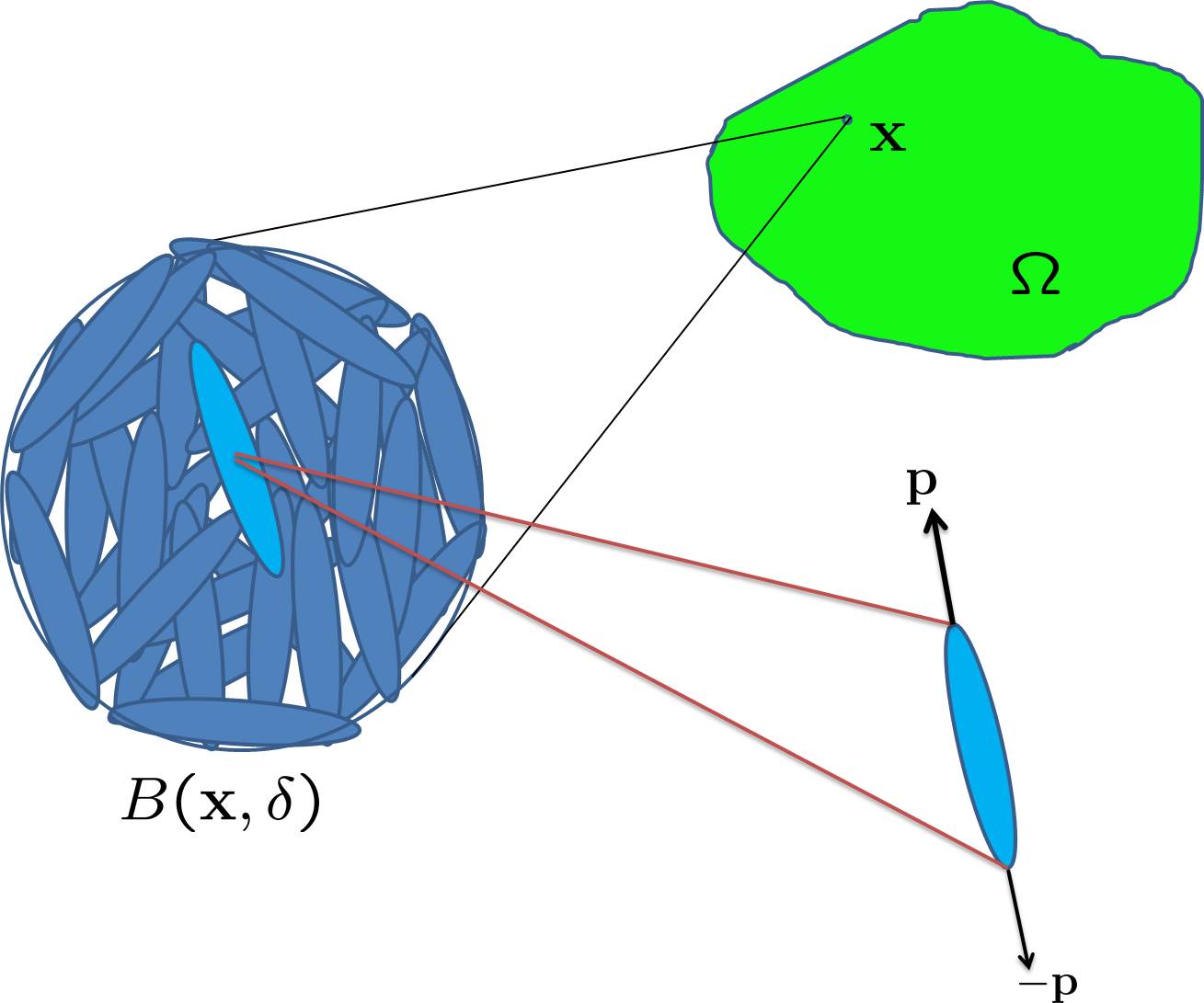}
  \caption{Picking molecules at random from those within $B(\x,\delta)$. }
  \label{picking}
\end{figure} 
we obtain a probability measure on $\R P^2$ for  the orientations of molecules in $B(\x,\delta)$, or equivalently  a probability measure $\mu=\mu_{\x}$ on the unit sphere $S^2$,  given by
\be 
\label{prob}
\mu_x=\frac{1}{N}\sum_{i=1}^N \half\left(\delta_{\p_i}+ \delta_{-\p_i}\right),
\ee
where $\pm\p_i$ denotes the orientation of the $i$th molecule. Here we take the point of view that we sample the orientations of molecules at a fixed time; however the molecules are in thermal motion, and by averaging the resulting probability measure over a macroscopically small time interval the value of $\delta$ for an effective statistical description could be reduced. More generally we will consider $\mu_\x$ to be a (Borel) probability measure on $S^2$ satisfying the head-to-tail symmetry condition 
\be 
\label{headtotail}
\mu_\x(E)=\mu_\x(-E) \mbox{ for all }\mu_\x-\mbox{measurable }E\subset S^2.
\ee
   For the time being we suppress the dependence of $\mu_\x$ on $\x$, denoting it simply by $\mu$, and we write $\langle g(\p)\rangle=\int_{S^2}g(\p)\,d\mu(\p)$ for any scalar, vector, or tensor $g=g(\p)$.

As an example, the measure 
\be 
\label{perfectalignment}
\mu=\half\left(\delta_{\e}+ \delta_{-\e}\right)
\ee
represents a state of perfect alignment of the molecules parallel to the unit vector $\e$. Such a state of perfect alignment being unrealistic, we will often consider $\mu$ to be a   continuously distributed measure  $d\mu(\p)=\rho(\p)\,d\p$, where $d\p$ denotes  the  surface area element on $S^2$ and $\rho\in L^1(S^2)$, $\rho\geq 0$, $\int_{S^2}\rho(\p)\,d\p=1$, $\rho(\p)=\rho(-\p)$ for a.e. $\p\in S^2$, which we can think of as a good approximation to the empirical measure in \eqref{prob} for $N$ large.

If the orientation of molecules is equally distributed
 in all directions, we say that the distribution is {\it isotropic}, and then $\mu=\mu_0$, where
$$  d\mu_0(\p)=\frac{1}{4\pi}d\p,$$
for which $\rho(\p)=\frac{1}{4\pi}$.

A natural idea would be to use as an order parameter the probability measure $\mu=\mu_{\x}$.  However this represents an infinite-dimensional state variable at each point $\x$, so it makes sense to use instead a finite-dimensional approximation consisting of a  finite number of {moments} of $\mu$. Because of \eqref{headtotail}  the first moment vanishes:
$$\int_{S^2}\p\,d\mu(\p)=0.$$
The second moment 
$$\M=\int_{S^2}\p\otimes \p\,d\mu(\p)$$
is a symmetric non-negative tensor satisfying $\tr\M=1$. The second moment tensor of the isotropic distribution $\mu_0$, $d\mu_0=\frac{1}{4\pi}d\p$, is
$$ \M_0=\frac{1}{4\pi}\int_{S^2}\p\otimes \p\,d\p=\frac{1}{3}{\bf 1}$$
(since $\int_{S^2}p_1p_2\,d\p=0,\;\int_{S^2}p_1^2\,d\p=\int_{S^2}p_2^2\,d\p$ etc, 
and ${\rm tr}\, \M_0=1$.)  The {\it de Gennes $\Q$-tensor} 
$$\Q=\M-\M_0=\int_{S^2}\left(\p\otimes \p-\frac{1}{3}{\bf 1}\right)d\mu(\p)$$ thus
measures the deviation of $\M$ from its isotropic value, and $\Q=\Q^T,\; \tr \Q=\bzero,\; \Q\geq-\frac{1}{3}{\bf 1}$ (i.e. $(\Q+\frac{1}{3}\1)\e\cdot\e\geq 0$ for all $\e\in S^2$).
(Note that whereas by construction $\Q=\bzero$ if $\mu=\mu_0$,  $\Q=\bzero$ does not imply $\mu=\mu_0$.
 For example we can take
$\mu=\frac{1}{6}\sum_{i=1}^3(\delta_{\e_i}+\delta_{-\e_i}).$)

Let us define
\be 
\label{E}{\mathcal E}=\{\Q\in M^{3\times 3}:\Q=\Q^T, {\rm tr}\,\Q=0\}.
\ee
Then it is easily checked that $\mathcal E$ is a 5-dimensional subspace of $M^{3\times 3}$ with orthonormal basis $\{\vE_1,\vE_2,\vE_3,\vE_4,\vE_5\}$, where
\be 
\label{onb} \vE_1=\frac{1}{\sqrt 2}(\e_2\otimes\e_3+\e_3\otimes\e_2),&& \vE_2=\frac{1}{\sqrt 2}(\e_3\otimes\e_1+\e_1\otimes\e_3),\nonumber\\ &&\hspace{-.8in}\vE_3=\frac{1}{\sqrt 2}(\e_1\otimes\e_2+\e_2\otimes\e_1),\\
\vE_4=\frac{1}{\sqrt 2}(\e_1\otimes\e_1-\e_3\otimes\e_3),&& \vE_5=\frac{1}{\sqrt 2}(\e_2\otimes\e_2-\e_3\otimes\e_3)\nonumber.
\ee
It is sometimes convenient to express an arbitrary $\Q\in\mathcal E$ in terms of this basis,
writing
\be 
\label{expandQ}
\Q=\sum_{i=1}^5q_i\vE_i,
\ee
where $q_i=\Q\cdot\vE_i$.

Since $\Q\in\mathcal E$, $\Q$ has a spectral decomposition
$$\Q=\lambda_1\n_1\otimes \n_1+\lambda_2\n_2\otimes \n_2+\lambda_3\n_3\otimes \n_3,$$
where $\{\n_i\}$ is an orthonormal basis of eigenvectors of $\Q$ with
corresponding eigenvalues $\lambda_i=\lambda_i(\Q)$ satisfying
$\lambda_1+\lambda_2+\lambda_3=0$. Since $\Q\geq -\frac{1}{3}{\bf 1}$, each $\lambda_i\geq-\frac{1}{3}$ 
and hence 
\be 
\label{evconstraints}
-\frac{1}{3}\leq \lambda_i\leq \frac{2}{3}.
\ee
Conversely, if each  $\lambda_i\geq -\frac{1}{3}$ then $\M$ is the second moment tensor for some $\mu$, e.g. for
$$\mu=\sum_{i=1}^3\left(\lambda_i+\frac{1}{3}\right)\frac{1}{2}(\delta_{\n_i}+\delta_{-\n_i}).$$

We can order the eigenvalues as $$\lambda_{\rm min}(\Q)\leq\lambda_{\rm mid}(\Q)\leq\lambda_{\rm max}(\Q)$$ with
$\n_{\rm max}=\n_{\rm max}(\Q),\,\n_{\rm mid}=\n_{\rm mid}(\Q),\, \n_{\rm min}=\n_{\rm min}(\Q)$  corresponding orthonormal eigenvectors. If $\lambda_{\rm min}(\Q)=-\frac{1}{3}$ then  we have $\Q\n_{\rm min}\cdot \n_{\rm min} =-\frac{1}{3}$, and hence
$$\int_{S^2}(\p\cdot \n_{\rm min})^2d\mu(\p)=0,$$
so that
$\mu$ is supported on the great circle of $S^2$ 
perpendicular to $\n_{\rm min}$. In particular, if $\mu$ is continuously distributed then the inequalities in \eqref{evconstraints} are strict.
If  also $\lambda_{\rm max}(\Q)=\frac{2}{3}$, so that $\lambda_{\rm min}(\Q)=\lambda_{\rm mid}(\Q)=-\frac{1}{3}$, then
$$\M\n_{\rm max}\cdot \n_{\rm max}=\int_{S^2}(\p\cdot \n_{\rm max})^2d\p=1,$$
and hence 
$$\int_{S^2}|\p\otimes \p-\n_{\rm max}\otimes \n_{\rm max}|^2d\mu=0,$$
  so that $\mu=\frac{1}{2}(\delta_{\n_{\rm max}}+\delta_{-\n_{\rm max}})$. 

Recall that we defined the director $\n$ as being the mean orientation of molecules. We can express this by looking for the $\n\in S^2$ that minimize
$$\int_{S^2}|\p\otimes\p-\n\otimes\n|^2d\mu(\p)= 2\int_{S^2}(1-(\p\cdot\n)^2)\,d\mu(\p)=2\left(\frac{2}{3}-\Q\n\cdot\n\right).$$
Thus the minimizers are $\n=\pm\n_{\rm max}(\Q)$. 

If two eigenvalues of $\Q$ are equal then $\Q$ is said to be {\it uniaxial} and has the form
\be 
\label{uniaxial}
\Q=s\left(\n\otimes \n-\frac{1}{3}{\bf 1}\right),
\ee
where $ \n\in S^2$ and the {\it scalar order parameter} $s\in[-\frac{1}{2},1]$ (with $s\in(-\frac{1}{2},1)$ if $\n$ is continuously distributed). Otherwise 
$\Q$ is {\it biaxial}. Provided $s>0$ the maximum eigenvalue $\lambda_{\rm max}(\Q)=\frac{2}{3}s$ of a uniaxial $\Q$ has multiplicity one, so that the $\n$ in \eqref{uniaxial} can be identified up to sign with the director. If $\Q$ is biaxial then $\lambda_{\rm max}(\Q)$ again has multiplicity one, so that the director is also well defined. In fact it is   difficult to experimentally  observe $\Q$ that are not very close to uniaxial with a nearly constant value of $s$ (typically in the range $0.6-0.7$). We will see why this is to be expected later. 

In order to give a more direct interpretation of $s$, note that  
\begin{eqnarray*}
\Q\n\cdot \n&=&\frac{2s}{3} =\langle(\p\cdot \n)^2-\frac{1}{3}\rangle\\ &=&
\langle\cos^2\theta-\frac{1}{3}\rangle,
\end{eqnarray*}
where $\theta$ is the angle between $\p$ and $\n$. Hence
$$s=\frac{3}{2}\langle \cos^2\theta-\frac{1}{3}\rangle.$$

\begin{prop} 
\label{uniaxialcondn}The tensor 
$\Q\in\mathcal E$ is uniaxial with scalar order parameter $s$ if and only if
\be 
\label{unicondn}
|\Q|^2=\frac{2s^2}{3},\;\det\Q=\frac{2s^3}{27}.
\ee
\end{prop}
\begin{proof}
That conditions \eqref{unicondn} are necessary is an easy computation using the formula $\det(\1+\az\otimes\bz)=1+\az\cdot\bz$. Conversely, if \eqref{unicondn} holds then the eigenvalues $\lambda_i$ of $\Q$ satisfy
\begin{eqnarray*}
\lambda_1+\lambda_2+\lambda_3&=&0,\\
\lambda_1^2+\lambda_2^2+\lambda_3^2&=&\frac{2s^2}{3},\\
\lambda_1\lambda_2\lambda_3&=&\frac{2s^3}{27},
\end{eqnarray*}
from which it follows that $\lambda_1\lambda_2+\lambda_2\lambda_3+\lambda_3\lambda_1=-\half(\lambda_1^2+\lambda_2^2+\lambda_3^2)=-\frac{s^2}{3}$. Thus the characteristic equation for $\Q$ is
$$\lambda^3-\frac{s^2}{3}\lambda-\frac{2s^3}{27}=0,$$
which factorizes as
$$\left(\lambda+\frac{s}{3}\right)^2\left(\lambda-\frac{2s}{3}\right)=0.$$
Letting $\n$ be the eigenvector corresponding to the eigenvalue $\frac{2s}{3}$ we obtain \eqref{uniaxial}.
\end{proof}
\begin{cor}
 Necessary and sufficient conditions for $\Q\in\mathcal E$ to be uniaxial with scalar order parameter $s\in[-\half,1]$ are that 
\be 
\label{uni}
|\Q|^6=54(\det\Q)^2,\;\det\Q\in \frac{2}{27}[-{\textstyle\frac{1}{8}},1].
\ee
\end{cor}
\begin{proof}
\eqref{unicondn} holds for some $s\in[-\half,1]$ if and only if \eqref{uni} does.
\end{proof}

Thus for nematic liquid crystals we have various possible choices for the order parameter:
\begin{itemize}
  \item the probability density function $\rho$ ($\infty$-dimensional, used in Onsager and Maier-Saupe models),
  \item $\Q$ (5-dimensional, used in the Landau - de Gennes theory),
  \item the pair $(s,\n)$  (3-dimensional, Ericksen theory \cite{ericksen1991liquid}),
  \item $\n$ (2-dimensional, Oseen-Frank theory).
  \end{itemize}  
We discuss these choices and models in the following sections.

\section{The isotropic to nematic phase transition}
\label{isotropicnematic}
We discuss this \\ (a) for  models in which the order parameter is the probability density function $\rho=\rho(\p)$,\\
(b) for a model in which the order parameter is $\Q$.\\
In both cases we assume that the order parameter is independent of $\x$ and look for minimizers of a corresponding free energy.

\subsection {Description using the probability density function}
There are two classical models, the Onsager and Maier-Saupe models, in both of which 
the probability measure  $\mu$ is assumed to be
 continuously distributed with density $\rho=\rho(\p)\in L^1(S^2)$, $\rho\geq 0$, $\int_{S^2}\rho(\p)\,d\p=1$, $\rho(\p)=\rho(-\p)$ for a.e. $\p\in S^2$, and in which the bulk free energy per particle at temperature $\theta>0$ has the form 
\be 
\label{onsager}
I_\theta(\rho)= U(\rho)-\theta \eta(\rho),
\ee
where $\eta(\rho)$ is an entropy term given by
$$ \eta(\rho)=-k_B\int_{S^2}\rho(\p)\ln\rho(\p)\,d\p,$$
 $k_B$ is Boltzmann's constant, and $U$  is an interaction term given by 
$$U(\rho)=\half\int_{S^2}\int_{S^2}K(\p,\q)\rho(\p)\rho(\q)\,d\p\,d\q.$$
We assume that the kernel $K:S^2\times S^2\to\R$ is frame-indifferent, so that 
$$K(\vR\p,\vR\q)=K(\p,\q) \mbox{ for all }\vR\in SO(3),$$
which due to a result of Cauchy (see \cite[p. 29]{truesdellnoll}) holds if and only if
$$K(\p,\q)=k(\p\cdot \q)$$
for some $k:[-1,1]\to \mathbb R.$ In the mean-field Maier-Saupe theory $U(\rho)$ is an internal energy term with $k$   given by
\be 
\label{MS}k(\p\cdot \q)=2\kappa\left(\frac{1}{3}-(\p\cdot \q)^2\right),
\ee
where $\kappa$ is a constant independent of temperature. In the Onsager theory, which corresponds to the case of a suspension of liquid crystal molecules in a solvent,   $U(\rho)$  represents positional entropy, with 
\be 
\label{Onsager}k(\p\cdot \q)=2\kappa \sqrt{1-(\p\cdot \q)^2},
\ee
where   $\kappa$ is proportional to both the temperature and  concentration. Denoting by
\be 
\label{Qrho}\Q(\rho)=\int_{S^2}\left(\p\otimes \p-\frac{1}{3}{\bf 1}\right)\rho(\p)\,d\p 
\ee
the corresponding $\Q$-tensor, we have that
\begin{eqnarray*}
|\Q(\rho)|^2&=&\int_{S^2}\int_{S^2}\left(\p\otimes \p-\frac{1}{3}{\bf 1}\right)\cdot\left(\q\otimes \q-\frac{1}{3}{\bf 1}\right)\rho(\p)\rho(\q)\,d\p\,d\q\\
&=&\int_{S^2}\int_{S^2}\left((\p\cdot \q)^2-\frac{1}{3}\right)\rho(\p)\rho(\q)\,d\p\,d\q.
\end{eqnarray*}
Hence for the Maier-Saupe potential $U(\rho)=-\kappa |\Q(\rho)|^2$ and 
\be 
\label{MS1}
I_\theta(\rho)=k_B\theta\int_{S^2}\rho(\p)\ln\rho(\p)\,d\p-\kappa|\Q(\rho)|^2.
\ee

Critical points of $I_\theta$ are solutions of the corresponding Euler-Lagrange equation obtained formally by setting $\frac{d}{d\tau}I_\theta(\rho+\tau u)|_{\tau=0}=0$ for $u$ satisfying $\int_{S^2}u(\p)\,d\p=0$, namely
\be 
\label{EL}
k_B\theta\ln \rho(\p)+\int_{S^2}k(\p\cdot\q)\rho(\q)\,d\q=c,
\ee
where $c$ is a constant. One solution of \eqref{EL} is the isotropic state $\rho(\p)=\frac{1}{4\pi}$. As shown in Fatkullin \& Slastikov \cite{FatkullinSlastikov2005}, Liu, Zhang \& Zhang \cite{LiuZhangZhang2005}, for the Maier-Saupe kernel all solutions can be determined explicitly and have the axially symmetric form 
\be 
\label{MSsolns}
\rho(\p,\e)=\frac{1}{4\pi\int_0^1\exp(-\sigma z^2)\,dz}\exp(-\sigma (\p\cdot\e)^2),
\ee
where $\e\in S^2$, and $\sigma$  is a function of the dimensionless parameter $\alpha=\frac{2\kappa}{k_b\theta}$,  the solution with $\sigma=0$ corresponding to the isotropic state. Up to rotation (that is, making different choices of $\e$) there can be 1, 2 or 3 distinct solutions depending on the value of $\alpha$. There is a transcritical bifurcation from the isotropic state at $\alpha=\frac{15}{2}$. The situation for the Onsager kernel is more complicated, since there are infinitely many bifurcation points from the isotropic state. However, by using techniques of equivariant bifurcation theory, expansions in spherical harmonics and variational arguments Vollmer \cite{vollmer2016} (see also Wachsmuth \cite{wachsmuth}) shows that there is a transcritical bifurcation to an axially symmetric solution, together with rotations of it, at the least bifurcation point $\alpha=\frac{32}{\pi}$, and she establishes other properties of the set of solutions, though a complete understanding of this set remains open.

\subsection{Description using a $\Q$-tensor model.}
We suppose that for a homogeneous (that is $\x$-independent) configuration  the free energy per unit volume (the bulk energy density) is given by a function $\psi_B(\Q,\theta)$ defined for trace-free symmetric $\Q=\int_{S^2}(\p\otimes \p-\frac{1}{3}\1)\,d\mu(\p)$ and an interval of temperatures $\theta$.

Consider two observers,  the first using the Cartesian coordinate system $\x=(x_1,x_2,x_3)$, and the second using translated and rotated coordinates $\z=\bar \x+\vR(\x-\bar \x)$, where $\bar\x\in\R^3, \vR\in SO(3)$. We require that both observers measure the same temperature and  free-energy density, that is
$$\psi_B(\Q^*,\theta)=\psi_B(\Q,\theta),$$
where $\Q^*$ is the value of $\Q$ measured by the second observer. Since 
\begin{eqnarray*}
\Q^*&=&\int_{S^2}\left(\q\otimes \q-\frac{1}{3}{\bf 1}\right)d\mu(\vR^T\q)\\
&=&\int_{S^2}\left(\vR\p\otimes \vR\p-\frac{1}{3}{\bf 1}\right)d\mu(\p)\\
&=&\vR\int_{S^2}\left(\p\otimes \p-\frac{1}{3}{\bf 1}\right)d\mu(\p)\vR^T,
\end{eqnarray*}
we deduce that  $\Q^*=\vR\Q\vR^T$ and so obtain the frame-indifference (isotropy) condition
\be 
\label{fi}
\psi_B(\vR\Q\vR^T,\theta)=\psi_B(\Q,\theta) \mbox{ for all } \vR\in SO(3).
\ee
In order to characterize functions satisfying \eqref{fi} we make use of the following standard result, giving a proof for the convenience of the reader.
\begin{lem}
\label{isotropic}
A function $f(\Q)$ of a real, symmetric, $3\times 3$ matrix $\Q$ is isotropic, that is
\be 
\label{iso} f(\vR\Q\vR^T)=f(\Q)\mbox{ for all }\vR\in SO(3),
\ee
if and only if $f(\Q)=g(\tr\Q,\tr\Q^2,\tr\Q^3)$ for some function $g$, and if $f$ is a polynomial so is $g$.
\end{lem}
\begin{proof}
Suppose $f$ is isotropic. Choosing $\vR$ to diagonalize $\Q$ we see that \eqref{iso} is equivalent to 
\be 
\label{newfi}
f(\Q)=f(\diag(\lambda_1,\lambda_2,\lambda_3)):=h(\lambda_1,\lambda_2,\lambda_3)
\ee
 for a function $h$ of the eigenvalues $\lambda_i$ of $\Q$, and choosing $\vR$ so as to permute these eigenvalues we deduce  that $h$ is symmetric with respect to  permutations of the $\lambda_i$. Since the eigenvalues are the roots of the characteristic equation
\be 
\label{chareqn}
\lambda^3-(\tr\Q)\lambda^2+(\tr\cof\Q)\lambda-\det\Q=0,
\ee
where $\cof\Q$ denotes the cofactor matrix of $\Q$, and since the coefficients determine the roots up to an arbitrary permutation, it follows that $h$ is a function of these coefficients, namely
\begin{eqnarray}\nonumber
\tr\Q&=&\lambda_1+\lambda_2+\lambda_3,\\ \label{coeffs}
\tr\cof\Q&=&\lambda_1\lambda_2+\lambda_2\lambda_3+\lambda_3\lambda_1,\\
\det\Q&=&\lambda_1\lambda_2\lambda_3,\nonumber
\end{eqnarray}
and hence, on account of the formulae 
\begin{eqnarray*}
\tr\cof\Q&=&\half\left((\tr\Q)^2-\tr\Q^2\right),\\ \det \Q&=&\tr\Q^3-\frac{3}{2}\tr\Q\,\tr\Q^2+\half(\tr\Q)^3,
\end{eqnarray*}
 $f$ is a function of $\tr\Q,\tr\Q^2,\tr\Q^3$. The converse is obvious since each of  $\tr\Q,\tr\Q^2,\tr\Q^3$ is isotropic.

If $f$ is a polynomial, then so is $h$, and by the fundamental theorem of symmetric polynomials (see, for example, \cite[\S 10]{edwards1984}) $h$ is a polynomial in the coefficients  \eqref{coeffs}, so that $g$ is a polynomial.
\end{proof}

\begin{prop}
\label{fiprop}The bulk energy 
$\psi_B$ satisfies the frame-indifference condition \eqref{fi} if and only if
\be 
\label{firep}
\psi_B(\Q,\theta)=g(\tr \Q^2,\tr \Q^3,\theta)
\ee
for some function $g$. 
If, for a given temperature $\theta$, $\psi_B(\Q,\theta)$ is a polynomial in $\Q$  then $g(\tr \Q^2,\tr\Q^3 ,\theta)$ is a polynomial in $\tr \Q^2,\tr\Q^3$.
\end{prop}
\begin{proof}Apply Lemma \ref{isotropic} to the  function $\hat\psi_B(\Q,\theta)=\psi_B(\Q-\frac{1}{3}(\tr\Q)\1,\theta)$, which is isotropic.
\end{proof}
Note that $\tr\Q^4=\half(\tr\Q^2)^2$ when $\Q=\Q^T, \tr\Q=0$. Hence by Proposition \ref{fiprop} the most general frame-indifferent $\psi_B$ that is a quartic polynomial in $\Q$ is a linear combination of $1,\tr\Q^2,\tr\Q^3$ and $\tr\Q^4$ with coefficients depending on $\theta$.\footnote{Similarly, for a sixth order polynomial $\psi_B$ is a linear combination of $1,\tr\Q^2,\tr\Q^3, \tr\Q^2\tr\Q^3, (\tr\Q^2)^3, (\tr\Q^3)^2$; see, for example,  \cite{gramsbergenetal1986}.} 
Following de Gennes \cite{degennes1971}, Schophol \& Sluckin \cite{schopohlsluckin1987}, Mottram \& Newton \cite{mottramnewton} we consider the special case 
\be 
\label{quarticbulk}
\psi_B(\Q, \theta)=a(\theta){\rm tr}\,\Q^2-\frac{2b}{3}{\rm tr}\,\Q^3+c{\rm tr}\,\Q^4,
\ee
where $b>0,c>0$ are constants independent of $\theta$ and  $a(\theta)=\alpha(\theta-\theta^*)$ for constants $\alpha>0, \theta^*>0$.  Thus we have dropped the term that is independent of $\Q$, since this does not affect the minimizing $\Q$, and made the approximation that the coefficient of $\tr\Q^2$ is linear in $\theta$, while the coefficients of $\tr\Q^3, \tr\Q^4$ do not depend on $\theta$ (in fact the expansion \eqref{expansion} below suggests that these coefficients should be proportional to $\theta$, but this   affects the following calculation only by changing the predicted value of the nematic initiation temperature $\theta_{\rm NI}$). Setting $a=a(\theta)$ we can write \eqref{quarticbulk} in terms of the eigenvalues $\lambda_i$ of $\Q$ as
\be 
\label{quartic}\psi_B=a\sum_{i=1}^3\lambda_i^2-\frac{2b}{3}\sum_{i=1}^3\lambda_i^3+c\sum_{i=1}^3\lambda_i^4.
\ee
 A calculation shows that the critical points of  \eqref{quartic}  subject to the constraint $\sum_{i=1}^3\lambda_i=0$  have two $\lambda_i$ equal, so that 
$\lambda_1=\lambda_2=\lambda, \lambda_3=-2\lambda$ say, and that 
$$\lambda(a+b\lambda+6c\lambda^2)=0.$$
Hence $\lambda=0$ is always a critical point, while if $a\leq\frac{b^2}{24c}$ there are also critical points   $\lambda=\lambda_\pm$,  where
$$\lambda_\pm=\frac{-b\pm\sqrt{b^2-24ac}}{12c}.$$
For   a critical point we have that
$$\psi_B=6a\lambda^2+4b\lambda^3+18c\lambda^4,$$
which is negative when 
$$6a+4b\lambda+18c\lambda^2=3a+b\lambda<0.$$
Thus there is a critical point with $\psi_B<0$ provided $3a+b\lambda_-<0$, and a short calculation shows that this holds if and only if $a<\frac{b^2}{27c}$. 

Hence we find that there is a phase transformation from an isotropic fluid to a uniaxial nematic phase at the critical temperature
$\theta_{\rm NI}=\theta^*+\frac{b^2}{27\alpha c}$. If $\theta>\theta_{\rm NI}$ then the unique minimizer of $\psi_B(\cdot,\theta)$ is $\Q=0$. If $\theta<\theta_{\rm NI}$ then the minimizers are
$$\Q=s\left(\n\otimes \n-\frac{1}{3}{\bf 1}\right)\;\; \mbox{ for }\n\in S^2,$$
with scalar order parameter 
\be 
\label{scalarop}
s=\frac{b+\sqrt{b^2-24ac}}{4c}>0. 
\ee

\subsection{Satisfaction of the eigenvalue constraints}
\label{evconstraintsa}
Recall from \eqref{evconstraints} and the subsequent discussion that the eigenvalues of the $\Q$-tensor should satisfy the constraints
\be 
\label{evconstraintsb}-\frac{1}{3}\leq \lambda_{\rm min}(\Q)\leq \lambda_{\rm mid}(\Q)\leq\lambda_{\rm max}(\Q)\leq\frac{2}{3},
\ee
and that the cases when $\lambda_{\rm min}(\Q)=-\frac{1}{3}$ represent states of perfect ordering (which can be regarded as unphysical).
However for the quartic bulk potential the   minimizers $\Q$ of $\psi_B(\Q,\theta)$ do not in general satisfy \eqref{evconstraintsb}, e.g. for MBBA with experimentally measured coefficients, the scalar order parameter of the nematic state exceeds 1 for temperatures only $7^\circ$C below the nematic initiation temperature.

A natural way  to enforce the eigenvalue constraints (suggested by Ericksen \cite{ericksen1991liquid} in the context of his liquid crystal theory) is to suppose that 
\be 
\label{blowup}
\psi_B(\Q,\theta)\to\infty\mbox{ as } \lambda_{\rm min}(\Q)\to -\frac{1}{3}.
\ee
A method to derive a singular bulk-energy $\psi_B^s(\Q,\theta)$ satisfying this condition from the Onsager model with the Maier-Saupe potential was proposed by Katriel et al \cite{katrieletal} and developed by  Majumdar and the author in \cite{u9,j59}, to which the reader is referred for a detailed treatment.

Given $\Q\in\mathcal E$   satisfying $\lambda_{\rm min}(\Q)>-\frac{1}{3}$ we define
\begin{eqnarray}
\psi_B^s(\Q,\theta)&=&\inf_{\{\rho:\Q(\rho)=\Q\}}I_\theta(\rho)\nonumber\\
&=&k_B\theta\inf_{\{\rho:\Q(\rho)=\Q\}}\int_{S^2}\rho(\p)\ln\rho(\p)\,d\p-\kappa|\Q|^2,\label{mspsi}
\end{eqnarray}
where $I_\theta(\rho), \Q(\rho)$ are defined in \eqref{onsager}, \eqref{Qrho}. Thus we just need to understand how  to minimize
$$E(\rho)=\int_{S^2}\rho(\p)\ln\rho(\p)\,d\p$$
subject to $\Q(\rho)=\Q$. We seek a minimizer of $E$ in
$${\mathcal A}_\Q=\{\rho\in L^1(S^2): \rho\geq 0, \int_{S^2}\rho(\p)\,d\p=1, \Q(\rho)=\Q\}.$$
(Note that we don't impose the condition $\rho(\p)=\rho(-\p)$ since this turns out to be automatically satisfied by a minimizer.) It is not hard to check that ${\mathcal A}_\Q$ is nonempty, and then a routine use of the direct method of the calculus of variations, making use of the fact that $\rho\ln\rho$ is a strictly convex function of $\rho$ having superlinear growth, gives
\begin{thm}  $E$ attains a minimum at a unique $\rho_\Q\in{\mathcal A}_\Q$.
\end{thm}
The minimizer $\rho_\Q$ can be given semi-explicitly:
\begin{thm}
\label{explicitrhoQ}
Let $\Q$ have spectral decomposition $\Q=\sum_{i=1}^3\lambda_i\n_i\otimes\n_i$. Then
$$\rho_\Q(\p)=\frac{\exp(\mu_1p_1^2+\mu_2p_2^2+\mu_3p_3^2)}{Z(\mu_1,\mu_2,\mu_3)},$$
where $\p=\sum_{i=1}^3 p_i\n_i$ and
$$Z(\mu_1,\mu_2,\mu_3)=\int_{S^2}\exp(\mu_1p_1^2+\mu_2p_2^2+\mu_3p_3^2)\,d\p.$$
The $\mu_i$  (unique up to adding a constant to each) solve the equations
\begin{eqnarray*}
\frac{\partial \ln Z}{\partial \mu_i}=\lambda_i+\frac{1}{3}, \;\;i=1,2,3.
\end{eqnarray*}
\end{thm}
Theorem \ref{explicitrhoQ} can be proved by showing that $\rho_\Q$ satisfies the corresponding Euler-Lagrange equation, the $\mu_i$ appearing as Lagrange multipliers. However, this is a bit tricky because of the possibility that  $\rho_\Q$ is not bounded away from zero.  A quicker proof is to use a `dual' variational principle for $\mu=(\mu_1,\mu_2,\mu_3)$ (see Mead \& Papanicolaou \cite{MeadPapanicolaou1984}), to maximize over $\mu\in{\mathbb R}^3$ the function
$$J(\mu)=\sum_{i=1}^3\left(\lambda_i+\frac{1}{3}\right)\mu_i-\ln Z(\mu).$$
For $\Q$ as above let $f(\Q)=E(\rho_\Q)=\min_{\rho\in{\mathcal A}_\Q}E(\rho)$, so that 
$$\psi_B^s(\Q,\theta)=k_B\theta f(\Q)-\kappa|\Q|^2.$$
Hence the bulk free energy has the form
\be 
\label{singularpotential}
\psi_B^s(\Q,\theta)=k_B\theta\left(\sum_{i=1}^3\mu_i\left(\lambda_i+\frac{1}{3}\right)-\ln Z(\mu)\right)-\kappa\sum_{i=1}^3\lambda_i^2,
\ee
and the following result shows that it satisfies the condition \eqref{blowup}.
\begin{thm}
\label{fprops}
 $f$ is strictly convex in $\Q$ and $$C_1-\frac{1}{2}\ln \left(\lambda_{\rm min}(\Q)+\frac{1}{3}\right)\leq f(\Q)\leq C_2- \ln \left(\lambda_{\rm min}(\Q)+\frac{1}{3}\right)$$
for constants $C_1,C_2$.
\end{thm}
Thus we may extend the definition of the singular bulk potential to all of $\mathcal E$ 
by defining $\psi_B^s(\Q,\theta)$ to be $+\infty$ if $\lambda_{\rm min}(\Q)\leq -\frac{1}{3}$. Then $\psi_B^s(\cdot,\theta):{\mathcal E}\to \R\cup\{+\infty\}$ is continuous,  and $\psi_B^s(\Q,\theta)$ is finite if and only if $\lambda_{\rm min}(\Q)>-\frac{1}{3}$. 

Other predictions of this model (see \cite{u9}) include:
\begin{enumerate}
\item All stationary points of $\psi_B^s(\cdot,\theta)$ are uniaxial and a phase transition is predicted from the isotropic to a uniaxial nematic phase just as
in the quartic model.
\item Minimizers $\rho^*$ of $I_\theta(\rho)$ correspond to  minimizers over $\Q$
of $\psi_B^s(\Q,\theta)$. 
\item Near $\Q=0$ we have (see also Katriel et al \cite{katrieletal}) the expansion
\begin{eqnarray}
\frac{1}{\theta k_B}\psi_B^s(\Q,\theta)&=&\ln 4\pi +\left(\frac{15}{4}-\frac{\kappa}{2 \theta k_B}\right){\rm tr}\,\Q^2\label{expansion}\\ &&\hspace{1in}-\frac{75}{14}{\rm tr}\,\Q^3+\frac{3825}{784}{\rm tr}\,\Q^4 + \ldots,\nonumber
\end{eqnarray}
\noindent which predicts in particular the relation $b/c=.91..$ for the coefficients in \eqref{quarticbulk}.
\end{enumerate}
(For generalizations to singular potentials for general moment problems see Taylor \cite{Taylor2016}.)

\section{The Landau - de Gennes theory}
\label{LdG}For simplicity we work at a constant temperature $\theta$. Let $\Omega$ be a bounded domain in $\mathbb R^3$. At each point $\x\in\Omega$, we have a corresponding order parameter tensor $\Q(\x)$. We suppose that the material is described by a free-energy density $\psi(\Q,\nabla \Q,\theta)$, so that the total free energy is given by
\be 
\label{energy}
 I_\theta(\Q)=\int_\Omega \psi(\Q(\x),\nabla \Q(\x),\theta)\,d\x.
\ee
We write $\psi=\psi(\Q,\vD,\theta)$, where $\vD$ is a third order tensor.

\subsection{Frame-indifference and material symmetry} To determine the conditions for $\psi$  to be frame-indifferent, we consider as before two observers, one using the Cartesian coordinates $\x=(x_1,x_2,x_3)$ and the second using translated and rotated coordinates $\z=\bar \x+\vR(\x-\bar \x)$, where $\vR\in SO(3)$, and we  require that $$\psi(\Q^*(\z),\nabla_\z\Q^*(\z),\theta)=\psi(\Q(\x), \nabla_\x\Q( \x),\theta),$$
where $\Q^*(\z)$ is the value of $\Q$ measured by the second observer. Since $\Q^*(\bar \x)=\vR\Q(\bar \x)\vR^T$, 
\begin{eqnarray*}
\frac{\partial Q_{ij}^*}{\partial z_k}(\z)&=&\frac{\partial}{\partial z_k}(R_{il}Q_{lm}( \x)R_{jm})\\ &=& \frac{\partial}{\partial x_p}(R_{il}Q_{lm}R_{jm})\frac{\partial x_p}{\partial z_k}\\
&=& R_{il}R_{jm}R_{kp}\frac{\partial Q_{lm}}{\partial x_p}.
\end{eqnarray*}
Thus, for every $\vR\in SO(3)$,
\be 
\label{frame}  \psi(\Q^*,\vD^*,\theta)=\psi(\Q,\vD,\theta),
\ee
where $\Q^*=\vR\Q\vR^T$, $D_{ijk}^*=R_{il}R_{jm}R_{kp}D_{lmp}$. Such $\psi$ are
called {\it hemitropic}.

The requirement that
$$\psi(\Q^*(\z), \nabla_\z\Q^*(\z),\theta)=\psi(\Q( \x),\nabla_\x\Q( \x),\theta)$$
when $\z=\bar \x+\hat \vR(\x-\bar \x)$, where $\hat \vR={\bf 1}-2\e\otimes \e$, $|\e|=1$,
is a {\it reflection}, is a condition of material symmetry satisfied by nematics, but not cholesterics, whose molecules have a chiral nature. 

Since any $\vR\in O(3)$ can be written as $\hat \vR\tilde \vR$, where $\tilde \vR\in SO(3)$ and $\hat \vR$ is a reflection, for a nematic
$$
\psi(\Q^*,\vD^*,\theta)=\psi(\Q,\vD,\theta),$$
where $\Q^*=\vR\Q\vR^T,\;D^*_{ijk}=R_{il}R_{jm}R_{kp}D_{lmp}$ and $\vR\in O(3)$. Such $\psi$ are called {\it isotropic}.

\subsection{Bulk and elastic energies}
We can decompose $\psi$ as
\begin{eqnarray*}
\psi(\Q,\nabla \Q,\theta)&=&\psi(\Q,\bzero,\theta)+(\psi(\Q,\nabla \Q,\theta)-\psi(\Q,\bzero,\theta))\\
&=&\psi_B(\Q, \theta)+\psi_E(\Q,\nabla \Q,\theta)\\
&=& \mbox{ bulk energy }+\mbox{ elastic energy },
\end{eqnarray*}
so that $\psi_B(\Q, \theta)=\psi(\Q,\bzero, \theta)$.

We have already studied the properties of $\psi_B(\Q,\theta)$. Usually it is assumed that $\psi_E(\Q,\nabla \Q,\theta)$ is quadratic in $\nabla \Q$.
Examples of isotropic functions quadratic  in $\nabla \Q$ are the invariants $I_i=I_i(\Q,\nabla\Q)$ given by
\begin{eqnarray}
     &&I_1 =   Q_{ij,k}Q_{ij,k}, \;\;
    I_2 =   Q_{ij,j}  Q_{ik,k}, \nonumber\\
   &&I_3 =  Q_{ik,j}Q_{ij,k}, \;\;
    I_4 = Q_{lk} Q_{ij,l} Q_{ij,k}.\label{isotr}
\end{eqnarray}
The first three linearly independent invariants $I_1,I_2,I_3$ span the possible isotropic quadratic functions of $\nabla\Q$.  The invariant $I_4$ is one of 6 possible linearly independent cubic terms that are quadratic in $\nabla \Q$ (see  \cite{berremanmeiboom1984,longaetal1987,poniewierskisluckin1985,schieletrimper1983}). Note that $$I_2-I_3=(Q_{ij}Q_{ik,k})_{,j}-(Q_{ij}Q_{ik,j})_{,k}$$ is a null Lagrangian.
An example of a hemitropic, but not isotropic, function is 
$$I_5=\varepsilon_{ijk}Q_{il}Q_{jl,k}.$$
For the elastic energy we take
\be 
\label{elastic} \psi_E(\Q,\nabla \Q, \theta)=\half\sum_{i=1}^{{5}} L_i(\theta)I_i(\Q,\nabla\Q),
\ee
where the $L_i(\theta)$ are material constants, with $L_5(\theta)=0$ for nematics. 

To summarize, we assume that for nematics and cholesterics
\be 
\label{freeenergy}
\psi(\Q,\nabla \Q,\theta)=\psi_B(\Q,\theta)+\half\sum_{i=1}^{5} L_i(\theta)I_i(\Q,\nabla\Q),
\ee
where $\psi_B(\Q,\theta)$ has one of the forms discussed 
with $L_5(\theta)=0$ for nematics.

\section{The constrained Landau - de Gennes and Oseen-Frank theories}
\label{constrainedOF}
For small\footnote{Since the $L_i$ are not dimensionless, some care is required in interpreting what it means for them to be small (see Gartland \cite{gartland2015}).} $L_i=L_i(\theta)$  it is reasonable to consider a constrained theory in which we require $\Q$ to be uniaxial with a constant scalar order parameter $s=s(\theta)>0$, so that 
\be 
\label{unicon}\Q=s\left( \n\otimes   \n-\frac{1}{3}{\bf 1}\right), \;\;   \n\in S^2,
\ee
and we minimize $I_\theta(\Q)$ in \eqref{energy} with \eqref{freeenergy} subject to the constraint \eqref{unicon}.
(For  rigorous work studying whether and when this is justified see Majumdar \& Zarnescu \cite{majumdarzarnescu}, Nguyen \& Zarnescu \cite{nguyenzarnescu}, Bauman, Phillips \& Park \cite{baumanparkphillips2012}, Canevari \cite{canevari2016}.) Then  the bulk energy just depends on $\theta$, so we only have to consider the elastic energy
\be 
\label{tildeI}
\tilde I_\theta(\Q)=\int_\Omega \psi_E(\Q,\nabla \Q,\theta)\,d\x.\ee

Formally calculating $\psi_E$ in \eqref{elastic} in terms of $\n,\nabla \n$ using \eqref{unicon} we obtain up to an additive constant  the {\it Oseen-Frank energy functional}
\begin{eqnarray}\label{OF}    \hspace{-1in}&& I_\theta(\n) = \half\int_{\Omega}\left(K_{1}(\textrm{div}\,
\n)^{2} + K_{2}(\n\cdot \textrm{curl}\,\n{+q_0})^{2}  +
K_{3}|\n \times \textrm{curl}\,\n|^{2}\right.\\&&\hspace{1.7in} \left. +
(K_{2}+K_{4})(\textrm{tr}(\nabla \n)^{2} - (\textrm{div}
\n)^{2})\right)\,d\x,\nonumber
\end{eqnarray}
where the {\it Frank constants} $K_i=K_i(\theta)$, and $q_0=q_0(\theta)$ are given by
\begin{eqnarray}
K_1&=&2L_1s^2+L_2s^2+L_3s^2-\frac{2}{3}L_4s^3,\nonumber\\K_2&=&2L_1s^2-\frac{2}{3}L_4s^3,\nonumber\\K_3&=&2L_1s^2+L_2s^2+L_3s^2+\frac{4}{3}L_4s^3,\label{LstoKs}\\
K_4&=&L_3s^2,\nonumber \\q_0&=&-\frac{L_5s^2}{2K_2},\nonumber
\end{eqnarray}
 and $ q_0=0$ for nematics, $ q_0\neq 0$ for cholesterics. 

\section{Boundary conditions}
\label{bcs}
We consider various boundary conditions that can be imposed on a part $\partial\om_1\subset\partial\om$ of the boundary:
\subsection{Constrained Landau - de Gennes and Oseen-Frank}
(i) {\it Strong anchoring}. Here $\n$ is specified on $\partial\om_1$:
 $$\n(\x)=\pm\bar \n(\x), \;\x\in\partial\Omega_1,$$
where $\bar\n:\partial\om_1\to S^2$.
Special cases are:
\begin{enumerate}
\item   {\it Homeotropic} boundary conditions : $\bar \n(\x)=\pm\nnu(\x)$, where $\nnu(\x)$ denotes the unit  outward normal to $\partial\om$ at $\x$.
\item 
   {\it Planar} boundary conditions : $\bar \n(\x)\cdot\nnu(\x)=0$.
\end{enumerate}
(ii) {\it Conical anchoring}:
$$|\n(\x)\cdot \nnu(\x)|=\alpha(\x)\in[0,1],\; \x\in\partial\Omega_1,$$
where $\alpha(\x)$ is given. Special cases are:
\begin{enumerate}
 \item $\alpha(\x)=1$, which is the same as homeotropic.
 \item 
 $\alpha(\x)=0$ {\it planar degenerate} (or {\it tangent}), where the director $\n$ is required to be parallel to boundary but  the preferred direction is not prescribed.
\end{enumerate}
(iii) {\it No anchoring}: no condition on $\n$ on $\partial\Omega_1$. This is natural mathematically but seems difficult to realize in practice.\\
(iv) {\it Weak anchoring}: no boundary condition is explicitly imposed, but a surface energy term is added to the energy \eqref{tildeI} or \eqref{OF}, of the form
$$\int_{\partial\Omega_1}w(\x,\n)\,dS$$
where $w(\x,\n)=w(\x,-\n)$ is prescribed. For example, corresponding to strong anchoring we can choose 
\be 
\label{weaka}
w(\x,\n)=-\half K(\n(\x)\cdot \bar \n(\x))^2,
\ee
with $K>0$, formally recovering the strong anchoring condition in the limit $K\rightarrow\infty$. When $\bar\n(\x)=\nnu(\x)$, $w(\x,\n)$ in \eqref{weaka} is the Rapini-Papoular form \cite{rapinipapoular1969} of the anchoring energy. Note that $w(\x,\n)$ is well defined in the constrained Landau - de Gennes theory and and can be alternatively written in the form
$$w(\x,\n)=-\half K\left(s^{-1}\Q(\x)\bar\n(\x)\cdot\bar\n(\x)+\frac{1}{3}\right).$$
\subsection{ Landau - de Gennes} 
(i) Strong anchoring:
$$\Q(\x)=\bar \Q(\x),\;\x\in\partial\Omega,$$
where $\bar\Q$ is prescribed.\\
(ii) Weak anchoring: add to the energy $I_\theta(\Q)$ in \eqref{energy} a surface energy term
$$\int_{\partial\Omega}w(\x,\Q)\,dS.$$

\section{Orientability}
\label{orientability}
But is the derivation of the  Oseen-Frank theory from Landau - de Gennes given in Section \ref{constrainedOF} correct? The constrained Landau - de Gennes theory is invariant to changing $\n$ to $-\n$, but is the same true of Oseen-Frank? The issue here is whether a line field can be {\it oriented}, i.e. turned into a vector field by assigning an orientation at each point. If we don't care about the regularity of the vector field this can always be done by choosing an arbitrary orientation at each point.

For $s$ a nonzero constant and $\n\in S^2$ let
$$P(\n)=s\left(\n\otimes \n-\frac{1}{3}{\bf 1}\right),$$
and set 
$${\mathcal Q}=\left\{\Q\in M^{3\times 3}:\Q=P(\n) \mbox{ for some }\n\in S^2\right\}.$$
Thus $P:S^2\rightarrow {\mathcal Q}$. The operator $P$ provides us with a way of `unorienting' an $S^2$-valued vector field. Given $\Q\in W^{1,1}(\Omega, {\mathcal Q})$ we say  $\Q$ is {\it orientable} if we can write
$$\Q(\x)=s\left(\n(\x)\otimes \n(\x)-\frac{1}{3}{\bf 1}\right),$$
where $\n\in W^{1,1}(\Omega,S^2)$. 
In topological language this means that $\Q$ has a {\it lifting} to $W^{1,1}(\Omega,S^2)$. Note that if $\Q\in W^{1,1}(\Omega, {\mathcal Q})$ is orientable with lifting $\n$, then since $n_i\in L^\infty(\om)$,
$$Q_{ij,k}=s(n_in_{j,k}+n_jn_{i,k})\mbox{  a.e. }\x\in\om,$$
from which it follows since $|\n|=1$ that  
\be 
\label{Qformula}
Q_{ij,k}n_j=sn_{i,k}.
\ee
In particular,  if $\Q\in W^{1,p}(\Omega,{\mathcal Q})$ is orientable for some $p$ with $1\leq p\leq\infty$, then $\n\in W^{1,p}(\om,S^2)$.   
\begin{thm}[{\cite[Proposition 2]{j61}}]
\label{twoorientations}
An orientable $\Q$ has exactly two liftings.
\end{thm}
\begin{proof}
Suppose that $\n$ and $\tau \n$ both generate $\Q$ and belong to $W^{1,1}(\Omega,S^2)$, where $\tau^2(\x)=1$ a.e.. 
Let $C\subset\Omega$ be a cube with sides parallel to the coordinate axes. Let $x_2,x_3$ be such that the line $x_1\mapsto (x_1,x_2,x_3)$ intersects $C$. Let $L(x_2,x_3)$ denote the intersection. For a.e. such $x_2,x_3$ we have that $\n(\x)$ and $\tau(\x)\n(\x)$ are absolutely continuous in $x_1$ on $L(x_2,x_3)$. Hence $\n(\x)\cdot \tau(\x)\n(\x)=\tau(\x)$ is continuous in $x_1$, so that $\tau(\x)$ is constant on $L(x_2,x_3)$.

Let $\varphi\in C_0^\infty(C)$. 
Then by Fubini's theorem
$\int_C\tau\varphi_{,1}d\x =0,$
so that the weak derivative $\tau_{,1}$ exists in $C$ and is zero. Similarly the weak derivatives $\tau_{,2}, \tau_{,3}$ exist in $C$ and are zero. Thus $\nabla\tau=0$ in $C$ and hence $\tau$ is constant in $C$. Since $\Omega$ is connected, $\tau$ is constant in $\Omega$, and thus $\tau\equiv 1$ or $\tau\equiv -1$ in $\Omega$.
\end{proof}

It is easy to construct smooth line fields  in a non simply-connected domain $\om$ which are not orientable (for a rigorous proof of non-orientability for such an example, illustrated in Fig. \ref{recover} (a) below,  see \cite[Lemma 11]{j61}). However if $\Omega$ is simply-connected we have the following result.
\begin{thm}[{\cite[Theorem 2]{j61}}]
\label{orientabilitythm}
If $\Omega\subset\R^3$ is a bounded simply-connected domain   of class $C^0$ and  
$\Q\in W^{1,2}(\om,{\mathcal Q})$  then $\Q$ is orientable.
\end{thm}
{\it Thus in a simply-connected domain the constrained Landau - de Gennes and Oseen-Frank theories are equivalent.}

 The ingredients of the proof of Theorem \ref{orientabilitythm} are:
\begin{itemize} \item  By a classical argument a lifting is possible if $\Q$ is smooth and $\Omega$ is simply connected
  \item A theorem of  Pakzad \& Rivi\`ere \cite{PakzadRiviere} implies that if $\partial\Omega$ is smooth, then there is a sequence of smooth $\Q^{(j)}:\Omega\rightarrow {\mathcal Q}$ converging weakly to $\Q$ in $W^{1,2}(\om,{\mathcal Q})$. 
  \item We can approximate a simply-connected domain with boundary of class $C^0$ by ones that are simply-connected with smooth boundary (see \cite{j65}); this step, and the assumption that $\om$ is of class $C^0$,  can be avoided using an argument of Bedford \cite[Proposition 3]{bedfordcholesterics}.
  \item  Orientability is preserved under weak convergence.
  \end{itemize}
For a related topologically more general lifting result see Bethuel \& Chiron \cite{bethuelchiron}.

In order to show that the constrained Landau - de Gennes and Oseen-Frank theories can result in different predictions for non simply-connected domains, we consider a three-dimensional modification  of a two-dimensional example from \cite[Section 5]{j61}, with   more realistic boundary conditions. We denote by $\om_\delta\subset\R^3$, $\delta>1$, the stadium-like open set 
\begin{equation}\om_\delta=\{\x=(x_1,x_2,x_3):(x_1,x_2)\in M_\delta, |x_3|<1\},\label{Mdelta}
\end{equation}
\begin{figure}[tbp] 
  \centering
  \includegraphics[width=4.09in,height=2.82in,keepaspectratio]{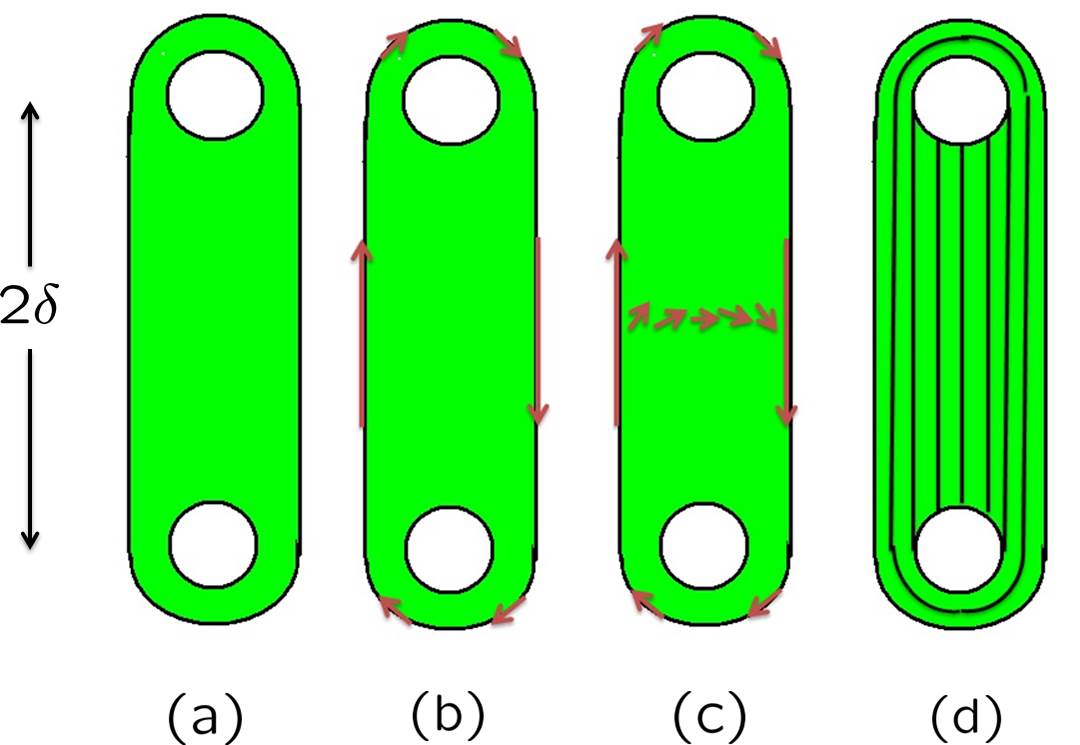}
  \caption{Cross-section (a) of stadium-like domain. The orientable outer boundary data $\n^+$ is shown in (b), and the idea for estimating the energy of any orientable director configuration in (c), with a nonorientable line field that has less energy than any orientable one for large $\delta$ shown in (d).}
  \label{stadium}
\end{figure}
shown in cross-section in Fig. \ref{stadium}(a), where $$M_\delta\stackrel{\rm def}{=}M_1\cup M_2\cup M_3\setminus( M_4\cup M_5)$$ and
\begin{equation}
\begin{array}{l}
M_1{=}\{\x=(x_1,x_2): x_1^2+(x_2-\delta)^2< 1\},\\
M_2{=}\{\x=(x_1,x_2): x_1^2+(x_2+\delta)^2< 1\},\\
M_3{=}\{\x=(x_1,x_2): |x_1|< 1; |x_2|\le \delta\},\\
M_4{=}\{\x=(x_1,x_2): x_1^2+(x_2-\delta)^2\le \frac{1}{4}\},\\
M_5{=} \{\x=(x_1,x_2): x_1^2+(x_2+\delta)^2\le\frac{1}{4}\}.
\label{M:domains}
\end{array}
\end{equation}
For simplicity we consider the constrained Landau - de Gennes energy in the one constant approximation $L_2=L_3=L_4=L_5=0,\, L_1>0$ with the following boundary conditions:
\begin{itemize}
  \item   on the curved  outer boundary $\partial(M_1\cup M_2\cup M_3)\times (-1,1)$ the line field is tangent to the boundary and lies in the $(x_1,x_2)-$plane,
  \item on the flat outer boundary $\{\x\in \partial \om_\delta:|x_3|=1\}$ the line field also lies in the $(x_1,x_2)-$plane,
  \item on the curved inner boundaries $(\partial M_4\cup\partial M_5)\times (-1,1)$ there is weak anchoring of Rapini-Papoular type.
  \end{itemize} 
Thus on the curved outer boundary we have planar boundary conditions in which the line field is specified, while on the flat outer boundary the boundary condition is planar degenerate. The corresponding energy  functional is
\be 
\label{functionalone}
I(\Q)=\half L_1\int_{\om_\delta}|\nabla\Q|^2d\x -\half K\int_{(\partial M_4\cup \partial M_5)\times (-1,1)} (\n\cdot\nnu)^2dS,
\ee
and we seek to minimize $I$ among  $\Q\in W^{1,2}(\om_\delta,{\mathcal Q})$  satisfying the boundary conditions.  
We first note that the set of such $\Q$ is nonempty, a member being given by  
$$\tilde \Q(\x)=\left\{\begin{array}{ll} s\left(\e_2\otimes \e_2-\frac{1}{3}\1\right),\, (x_1,x_2)\in M_3\cap M_\delta\\
s\left(\n_\delta(\x)\otimes \n_\delta(\x)-\frac{1}{3}\1\right),\, (x_1,x_2)\in M_1\setminus M_4,\, x_2\geq\delta\\
s\left(\m_\delta(\x)\otimes \m_\delta(\x)-\frac{1}{3}\1\right),\,(x_1,x_2)\in M_2\setminus M_5,\, x_2\leq -\delta
\end{array}\right.$$ 
where $\e_2=(0,1,0)$ and
$$\n_\delta(\x){=}\left(\frac{x_2-\delta}{|(x_1,x_2-\delta)|},-\frac{x_1}{|(x_1,x_2-\delta)|},0\right),$$
$$\m_\delta(\x){=}\left(\frac{x_2+\delta}{|(x_1,x_2+\delta)|},-\frac{x_1}{|(x_1,x_2+\delta)|},0\right),$$
the corresponding line field being illustrated in Fig. \ref{stadium}(d). Since the constraint $\Q=s\left(\n\otimes\n-\frac{1}{3}\1\right)$ is closed with respect to weak convergence in $W^{1,2}(\om_\delta,M^{3\times 3})$, and the embedding of $W^{1,2}(\om_\delta,M^{3\times 3})$ in $L^2(\partial \om_\delta,M^{3\times 3})$ is compact, a routine use of the direct method of the calculus of variations shows that $I$ attains a minimum on $W^{1,2}(\om_\delta,{\mathcal Q})$ subject to the boundary conditions. We will show that if $\delta$ is sufficiently large then any minimizer is non-orientable, even though the corresponding Oseen-Frank functional has a minimizer.

From \eqref{LstoKs} we see that the corresponding Oseen-Frank functional is
\be 
\label{functionaltwo}
I(\n)= s^2L_1\int_{\om_\delta}|\nabla\n|^2d\x -\half K\int_{(\partial M_4\cup \partial M_5)\times (-1,1)} (\n\cdot\nnu)^2dS,
\ee
which is to be minimized for $\n\in W^{1,2}(\om_\delta,S^2)$ subject to the boundary condition   $\n\cdot\e_3=0$ on the outer boundary. The set of such $\n$ is also nonempty. To see this it suffices to consider $\n$ of the form $\n=(\cos\theta(x_1,x_2),\sin\theta(x_1,x_2),0)$ with $(x_1,x_2)\in M_\delta$. Let $S=M_1\cup M_2\cup M_3$ and define $\n_\pm:\partial S\to S^2$ by
\be 
\label{nplusminus}
\n_\pm(\x)=\left\{\begin{array}{ll}
\pm\e_2 &\mbox{if }x_1=-1, |x_2|<\delta,\\
\pm(x_2-\delta,-x_1,0)&\mbox{if } x_1^2+(x_2-\delta)^2=1, x_2\geq \delta,\\
\mp\e_2&\mbox{if }x_1=1, |x_2|<\delta,\\
\pm(x_2+\delta,-x_1,0)&\mbox{if }x_1^2+(x_2+\delta)^2=1, x_2\leq-\delta,
\end{array}\right.
\ee
so that $\n_+$ is as shown in Fig. \ref{stadium}(b), with $\n_-$ the corresponding anti-clockwise unit vector field. Note that since $S$ is convex, its gauge function with respect to the interior point $(0,\delta)$, namely
$$f(x_1,x_2)=\inf\{t>0:(x_1,x_2-\delta)\in tS\}$$ is convex and hence Lipschitz (see, for example, \cite{rockafellar70}). Hence 
$$\tilde\n(\x)=\n_+\left(f(x_1,x_2)^{-1}\left(x_1,x_2-\delta\right),0\right),$$
as the composition of Lipschitz maps, is a Lipschitz map from $\om_\delta$ to $S^2$, and hence belongs to $W^{1,2}(\om_\delta,S^2)$ and satisfies $\tilde\n\cdot\e_3=0$ and is tangent to the boundary on $(\partial S\times (-1,1))\cup \{\x\in \partial\om_\delta:|x_3|=1\}$. 

Now suppose that $\n\in W^{1,2}(\om_\delta,S^2)$ is orientable and satisfies the boundary conditions. Then the trace of $\n$ belongs to $W^{\half,2}(\partial S\times (-1,1),S^2)$, and thus $\n\cdot(\n_+,0)\in W^{\half,2}(\partial S\times(-1,1))$ and takes only the values $\pm 1$. By \cite[Theorem B.1]{BourgainBrezisMironescu2000} (see also \cite[Lemma 9]{j61})  $\n\cdot(\n^+,0)$ is constant on $\partial S\times(-1,1)$, so that either $\n=(\n_+,0)$ on $\partial S\times(-1,1)$ or $\n=(\n_-,0)$ on $\partial S\times(-1,1)$. Let us suppose that $\n=(\n^+,0)$ on $\partial S\times(-1,1)$, the other case being treated similarly. For $|x_2|<\delta-\half, |x_3|<1$ consider the line segment $J=\{(x_1,x_2,x_3):|x_1|\leq 1\}$. For a.e. such $x_2,x_3$ we have that $\n|_J\in W^{1,2}(J,S^2)$ with $\n(-1,x_2,x_3)=\e_2,\; \n(1,x_2,x_3)=-\e_2$. The minimum value of $\int_{-1}^1|\m_{,1}|^2dx_1$ among $\m\in W^{1,2}((-1,1),S^2)$ satisfying $\m(-1)=\e_2,\;\m(1)=-\e_2$ is easily checked to be $\pi^2/2$. Therefore 
\be 
\label{stadest}
I(\n)\geq 2\pi^2s^2L_1\left(\delta-\half\right)-2\pi K,
\ee
since the sum of the surface areas of the two inner cylinders is $4\pi$. But
\be 
\label{stadest1}
I(\tilde\Q)\leq 4\pi s^2L_1\int_\half^1r^{-1}dr= 4\pi\ln 2\, s^2\;L_1,
\ee
so that from \eqref{stadest} any minimizer of $I(\Q)$ subject to the boundary conditions is non-orientable if 
$$2\pi^2s^2L_1\left(\delta-\half\right)-2\pi K>4\pi\ln 2\, s^2\; L_1,$$
that is if
\be 
\label{largedelta}
\delta >\half+\frac{2\ln 2}{\pi}+\frac{K}{\pi s^2L_1}.
\ee
 
\section{Existence of minimizers in the general Landau - de Gennes theory.}
\label{existenceLdG}
 Using the direct method of the calculus of variations one can prove
\begin{thm}[Davis \& Gartland \cite{gartlanddavis}]
Let $\Omega\subset\mathbb{R}^3$ be a bounded domain with Lipschitz boundary $\partial\Omega$.   For fixed $\theta>0$ let $\psi_B(\cdot,\theta)$ be continuous and bounded below on ${\mathcal E}$, and assume the constants $L_i=L_i(\theta)$ satisfy $L_4=L_5=0$ and 
\be 
\label{longa}
L_1>0, -L_1<L_3<2L_1,  L_1+\frac{5}{3}L_2+\frac{1}{6}L_3>0.
\ee
Let $\bar \Q:\partial \Omega\rightarrow \mathcal{E}$ belong to $W^{\half,2}(\om,\mathcal E)$, where $\mathcal E$ is defined in \eqref{E}. Then 
$$I_\theta(\Q)= \int_\Omega\left(\psi_B(\Q,\theta)+\half\sum_{i=1}^3L_i(\theta) I_i(\nabla \Q)\right)\,dx$$
attains a minimum on 
$$\mathcal{A}=\{\Q\in W^{1,2}(\Omega,\mathcal{E}):\Q|_{\partial\Omega}=\bar \Q\}.$$
\end{thm}
\begin{rem}\label{longarem}\rm The inequalities \eqref{longa} are necessary and sufficient in order that $$\psi_E(\nabla\Q,\theta)=\half\sum_{i=1}^3 L_iI_i(\nabla\Q)\geq c|\nabla\Q|^2$$ for all $\Q\in\mathcal E$ and some $c>0$, and thus for the quadratic function $\psi_E(\cdot,\theta)$ to be strictly convex (see, for example, \cite{longaetal1987}).
\end{rem}
\begin{rem}\rm
The result holds also for $L_5\neq 0$ provided we make the slightly stronger assumption on $\psi_B$ that 
$$\psi_B(\Q,\theta)\geq c_0(\theta)|\Q|^p-c_1(\theta)$$ for all $\Q\in\mathcal E$, where $p>2, c_0(\theta)>0$ and $c_1(\theta)$ are constants, using the fact that for any $\ep>0$ there is a constant $C_\ep(\theta)>0$ such that  $$\int_\om|I_5|\,d\x\leq \ep\int_\om(|\nabla\Q|^2+|\Q|^p)\,d\x+C_\ep(\theta)$$
for all $\Q\in W^{1,2}(\om,\mathcal E)$.
\end{rem} 
In the case of the quartic bulk potential \eqref{quarticbulk} Davis \& Gartland \cite{gartlanddavis} used elliptic regularity to show that any minimizer $\Q^*$ is a smooth solution of the corresponding  Euler-Lagrange equation, given in weak form by 
\be 
\label{ELeqn}
\int_\om \left(\frac{\partial \psi}{\partial\Q}\cdot \vP+\frac{\partial\psi}{\partial\nabla\Q}\cdot\nabla\vP\right)\,d\x=0
\ee
for all $\vP\in C_0^\infty(\om,{\mathcal E})$. Note that \eqref{ELeqn} is semilinear elliptic in the variables $\q=(q_1,\ldots,q_5)$ given in \eqref{expandQ},  on account of $\psi_E$ not depending explicitly on $\Q$, so that $\psi_E(\nabla\Q,\theta)\geq c|\nabla\q|^2$. 

The hypothesis of Proposition \ref{unbdd} that  $L_4=0$ is unsatisfactory because it  implies by \eqref{LstoKs} that $K_1=K_3$, which is not generally true.  But  if $L_4\neq 0$ we have
\begin{thm}[\cite{u9,j59}]\label{unbdd}
 For any boundary conditions, if $\psi_B(\cdot,\theta)$ is real-valued, continuous and bounded below, and if $L_4(\theta)\neq 0$ then 
$$I_\theta(\Q)=\int_\Omega\left(\psi_B(\Q,\theta)+\half\sum_{i=1}^4L_i(\theta)I_i(\Q,\nabla\Q)\right)\,dx$$
is unbounded below.
\end{thm}
If we use the singular bulk energy $\psi_B^s(\Q,\theta)$ defined in \eqref{singularpotential}, which does not satisfy the hypotheses on $\mathcal E$ in Theorem \ref{unbdd} on account of Theorem \ref{fprops}, then we can prove existence  when $L_4\neq 0$ under suitable inequalities on the $L_i$, because then the eigenvalue constraint $\lambda_{\rm min}(\Q(\x))>-\frac{1}{3}$ is satisfied a.e. in $\om$ whenever $I_\theta(\Q)<\infty$. For example, if $L_4>0$ then  for any $\Q\in \mathcal E$ with $\lambda_{\rm min}(\Q)\geq -\frac{1}{3}$ we have that 
$$L_4I_4(\Q,\nabla\Q)=L_4Q_{lk} Q_{ij,l} Q_{ij,k}\geq -\frac{1}{3}L_4|\nabla \Q|^2.$$ 
Hence if the inequalities
\be 
\label{longa1}
L_1'>0, -L_1'<L_3<2L_1',  L_1'+\frac{5}{3}L_2+\frac{1}{6}L_3>0.
\ee
hold with $L_1'=L_1-\frac{1}{3}L_4$, then by Remark \ref{longarem} the elastic energy $\psi_E(\Q,\nabla\Q,\theta)=\sum_{i=1}^4I_i(\Q,\nabla\Q)$ is coercive and strictly convex in $\nabla\Q$, and we can apply the direct method in a straightforward way to prove the existence of a minimizer $\Q^*$.

But now it is not so obvious that the Euler-Lagrange equation \eqref{ELeqn} holds, because of the one-sided constraint $\lambda_{\rm min}(\Q(\x))>-\frac{1}{3}$ for a.e. $\x\in\om$. In order to  prove that a minimizer $\Q^*$ is a weak solution of \eqref{ELeqn} it is natural to first try to show that $\lambda_{\rm min}(\Q^*(\x))$ is bounded away from $-\frac{1}{3}$, that is   
\be 
\label{unifbd}
\lambda_{\rm min}(\Q^*(\x))\geq -\frac{1}{3}+\delta \mbox{ for a.e. }\x\in\om \mbox{ and some }\delta>0, 
\ee
because then we can construct two-sided variations. We might expect this to be true because otherwise the integrand will be unbounded in the neighbourhood of some point of $\om$. However many examples from the calculus of variations, even in one dimension (see \cite{j25,j28}), show that it can indeed happen that minimizers have unbounded integrands (tautologically because having the integrand infinite somewhere can enable it to be smaller somewhere else).

{\it  It is an open problem to prove \eqref{unifbd}  for general elastic constants }$L_i$.\footnote{A related, and even harder, open problem is that of proving that  minimizers $\y^*:\om\to\R^3$ of the elastic energy $I(\y)=\int_\om W(\nabla\y(\x))\,d\x$ in nonlinear elasticity under the non-interpenetration hypothesis $W(\A)\to \infty$ as $\det\A\to 0+$ satisfy $\det \nabla y^*(\x)\geq\delta>0$ a.e. in $\om$.}  However, in the one-constant case this can be proved:
\begin{thm}[\cite{u9,j59}]
\label{boundedaway}
Let $\Q^*$ minimize
$$I_\theta(\Q)=\int_\Omega \left(\psi_B^s(\Q,\theta)+\half L_1(\theta)|\nabla \Q|^2\right)\,dx,$$
subject to $\Q|_{\partial\om}=\Q_0$, where $L_1(\theta)>0$ and $\Q_0(\cdot)$ is sufficiently smooth with $\lambda_{\min}(\Q_0(\x))>-\frac{1}{3}$. Then
$$\lambda_{\min}(\Q^*(x))>-\frac{1}{3}+\delta,$$
for some  $\delta>0$ and $\Q^*$ is a smooth solution of  the Euler-Lagrange equation \eqref{ELeqn}.
\end{thm}
For additional partial results see Evans, Kneuss \& Tran \cite{EvansKneussTran2016} and Bauman \& Phillips \cite{baumanphillips2016}. 

\section{Description of defects}
\label{defects}
\subsection{Summary of liquid crystal models}
From now on we shall for simplicity restrict attention to nematics and drop the explicit dependence on the temperature. Thus we consider the Landau - de Gennes energy functional
\be 
\label{LdGf}
I_{\rm LdG}(\Q)=\int_\om \psi(\Q,\nabla\Q)\,d\x,
\ee 
where $\psi(\Q,\nabla\Q)=\psi_B(\Q)+\psi_E(\Q,\nabla\Q)$,  $\psi_B(\Q)$ has one of the forms \eqref{quarticbulk}, \eqref{singularpotential} previously discussed, and  $\psi_E(\Q,\nabla\Q)=\half\sum_{i=1}^4L_iI_i(\Q,\nabla\Q)$, together with the corresponding Oseen-Frank energy functional
\be 
\label{OFf}
I_{\rm OF}(\n)=\int_\om W(\n,\nabla\n)\,d\x,
\ee
where 
\begin{eqnarray}
\label{OFe}    2W(\n,\nabla\n)  = K_{1}(\textrm{div}\,
\n)^{2} + K_{2}(\n\cdot \textrm{curl}\,\n)^{2} \\&& \hspace{-1.5in} +
K_{3}|\n \times \textrm{curl}\,\n|^{2}  +
(K_{2}+K_{4})(\textrm{tr}(\nabla \n)^{2} - (\textrm{div}
\n)^{2}).\nonumber
\end{eqnarray}
The Euler-Lagrange equation for $I_{\rm OF}(\n)$ subject to the pointwise constraint $|\n(\x)|=1$ is given by 
\be 
\label{ELOF}
({\bf 1}-\n\otimes\n)\left({\rm div}\,\frac{\partial W}{\partial \nabla \n}-\frac{\partial W}{\partial \n}\right)=\bzero,
\ee
or equivalently by
\be 
\label{ELOF1}
\Div\frac{\partial W}{\partial\nabla\n}-\frac{\partial W}{\partial\n}=\lambda(\x)\n(\x),
\ee
where $\lambda(\x)$ is a Lagrange multiplier.

As we have seen, under the pointwise uniaxial constraint $$\Q=s\left(\n\otimes\n-\frac{1}{3}\1\right)$$ with $s>0$ constant the two functionals \eqref{LdGf}, \eqref{OFf} in general give different predictions, whereas they are equivalent for simply-connected domains $\om$. Another possible ansatz is to allow $s$ to depend on $\x$, when $I_{\rm LdG}$ reduces to an energy functional of the form proposed by Ericksen \cite{ericksen1991liquid}
\be 
\label{ericksenenergy}
I_{\rm E}(s,\n)=\int_\om W(s, \nabla s, \n, \nabla\n)\,d\x.
\ee
 \subsection{Function spaces}
As described, for example, in \cite{j70,j67,bedfordcholesterics}, it is not sufficient to specify the energy functional, as part of the model is also the function space in which minimizers (and critical points etc) are to be sought. The larger this function space, the wilder potential singularities of minimizers and critical points may be. Changing the function space can change the predicted minimizers, as well as the minimum value of the energy (the {\it Lavrentiev phenomenon}) as described both for nonlinear elasticity and liquid crystals in \cite{j70}. 

As we have already seen, the usual function space considered for $\Q$ in the Landau - de Gennes functional $I_{\rm LdG}(\Q)$ is the Sobolev space $W^{1,2}(\om,\mathcal E)$. For the Oseen-Frank energy $I_{\rm OF}(\n)$ we   have that $W(\n,\nabla\n)\leq c_1|\nabla\n|^2$ for some constant $c_1>0$, while the inequality
\be 
\label{coercive} W(\n,\nabla\n)\geq c_0|\nabla\n|^2
\ee 
for some constant $c_0>0$ holds if and only if the 
 {\it Ericksen inequalities}
 \cite{ericksen1966}
\be 
\label{ericksenineq}
K_1>0,\, K_2>0,\, K_3>0,\, K_2>|K_4|, \,2K_1>K_2+K_4,
\ee
 are satisfied. Hence a natural function space for $\n$ is $W^{1,2}(\om,S^2)$.  
\subsection{Point defects}
Defects can roughly be thought of as locations in the neighbourhood of which the order parameter ($\Q-$tensor, director or line field) changes rapidly. How they are described depends on the model and function space used. The simplest {\it point defect}, located at the point $\x=\bzero\in\om$, is described by the radial hedgehog with director field
\be 
\label{hedgehog}
\hat\n(\x)=\frac{\x}{|\x|}.
\ee
For $\x\in \om\setminus\bzero$ and any Frank constants the hedgehog is a smooth solution of \eqref{ELOF} with gradient
\be 
\label{hedgehoggradient}
\nabla\hat\n(\x)=\frac{1}{|\x|}\left(\1-\hat\n(\x)\otimes\hat\n(\x)\right),
\ee
so that $|\nabla\hat\n(\x)|=\frac{\sqrt 2}{|\x|}$. After checking that indeed \eqref{hedgehoggradient}   gives the weak derivative of $\hat\n$, we see that $\hat\n\in W^{1,p}(\om,S^2)$ if and only if $1\leq p<3$, so that $I_{\rm OF}(\hat\n)<\infty$.

In  the {\it one-constant approximation} $K_1=K_2=K_3=K,\; K_4=0$ the hedgehog $\hat\n$ is the unique
 minimizer of $I_{OF}(\n)=\half K\int_\Omega|\nabla\n|^2d\x$ subject to its own boundary conditions
(see Brezis, Coron \& Lieb \cite{breziscoronlieb}, Lin \cite{lin1987}). In the one-constant  case any minimizer is smooth in $\om$ 
except possibly for a finite number of point defects 
(Schoen \& Uhlenbeck \cite{schoenuhlenbeck}) at points $\x(i)\in\om$ such that
$$\n(\x)\sim \pm\vR(i)\frac{\x-\x(i)}{|\x-\x(i)|}\mbox{ as }\x\to\x(i),$$
for some $\vR(i)\in SO(3)$. 

For general elastic constants $K_i$ it is not known whether minimizers can only have a finite number of point defects, though by a partial regularity result of Hardt,  Kinderlehrer \& Lin \cite{hardtkinderlehrerlin1988} the set of singularities has one-dimensional Hausdorff measure zero. The conditions under which the hedgehog minimizes $I_{\rm OF}(\n)$ subject to its own boundary conditions are not known.  H\'elein \cite{helein} observed that the method of Lin \cite{lin1987} shows that the hedgehog is energy-minimizing if $K_2\geq K_1$, a detailed proof being given by Ou \cite{ou92}. At the same time, the work of H\'elein \cite{helein}, Cohen \& Taylor \cite{cohentaylor} and Kinderlehrer \& Ou \cite{KinderlehrerOu1992} established that the second variation of $I_{\rm OF}(\n)$ at $\hat\n$ is positive if and only if $8(K_2-K_1)+K_3\geq 0$. Thus $\hat\n$ is not minimizing if $8(K_2-K_1)+K_3<0$. For more discussion see \cite{j70}.

One indication as to why the one-constant approximation is easier than the general case is that in general the Lagrange multiplier $\lambda(\x)$ corresponding to the pointwise constraint $|\n(\x)|=1$ in the Euler-Lagrange equation \eqref{ELOF1} for $I_{\rm OF}$
in general depends on second derivatives of $\n$, as can be seen by taking the inner product of \eqref{ELOF1} with $\n$. However the identity $\Delta\n\cdot\n=-|\nabla\n|^2$ for $|\n|=1$ shows that in the one-constant case $\lambda =K|\nabla\n |^2$ is an explicit function of $\nabla\n$.

Since weak solutions in the Landau - de Gennes are smooth, modulo the 
 difficulties with the eigenvalue constraints described in Section \ref{existenceLdG} when the singular bulk potential is used,  defects are not represented
by   singularities in $\Q$. Hence the best way to characterize defects is unclear (for a discussion see Biscari \& Peroli \cite{BiscariPeroli1997}). 
 In both the Landau - de Gennes
 and Ericksen theories there are solutions to the Euler-Lagrange
equations representing {\it melting hedgehogs}, of the form
$$\Q(\x)=s(|\x|)\left(\frac{\x}{|\x|}\otimes\frac{\x}{|\x|}-\frac{1}{3}{\bf 1}\right),$$
where $s(0)=0$. For the quartic bulk energy $\psi_B$ and the one constant elastic energy such a solution is 
shown by   Ignat, Nguyen,  Slastikov \& Zarnescu \cite{ignatetal2015} to be a local minimizer for $\Omega={\mathbb R}^3$ 
subject to the 
condition at infinity
$$\Q(\x) \to s\left(\frac{\x}{|\x|}\otimes \frac{\x}{|\x|}-\frac{1}{3}
{\bf 1}\right) \mbox{ as } |\x|\to \infty,$$
where
$s=\frac{b+\sqrt{b^2-24ac}}{4c}>0$, for temperatures
close to the nematic initiation temperature. However for lower temperatures the melting hedgehog is not  a minimizer 
(Gartland \& Mkaddem \cite{GartlandMkaddem1999}) and numerical evidence suggests
a biaxial torus structure for the defect without melting. For other work on the description of the hedgehog defect according to the Landau - de Gennes theory see, for example, \cite{HenaoMajumdar2012,HenaoMajumdarPisante2016,KraljVirga2001,lamy2013,Majumdar2012}. 

The situation as regards minimizers in the Landau - de Gennes theory being smooth might be different for free-energy densities $\psi(\Q,\nabla \Q)$
which are convex but not quadratic in $\nabla \Q$. For such integrands there is a counterexample of
\v{S}ver\'{a}k \& Yan \cite{sverakyan00} with a singular minimizer of the form
$$\Q(\x)=|\x|\left(\frac{\x}{|\x|}\otimes\frac{\x}{|\x|}-\frac{1}{3}{\bf 1}\right).$$

\subsection{Line defects}
It is natural to consider a two-dimensional version of the hedgehog given by the director field (see Fig. \ref{2Dhedgehog})
$$\tilde \n(\x)=\left(\frac{x_1}{r},\frac{x_2}{r},0\right),\;\;r=\sqrt{x_1^2+x_2^2}\,,$$
defined for $\x$ belonging to the cylinder $\om=\{\x:0<x_3<L, r<1\}$.
\begin{figure}[htbp] 
  \centering
  \includegraphics[width=4.09in,height=2.18in,keepaspectratio]{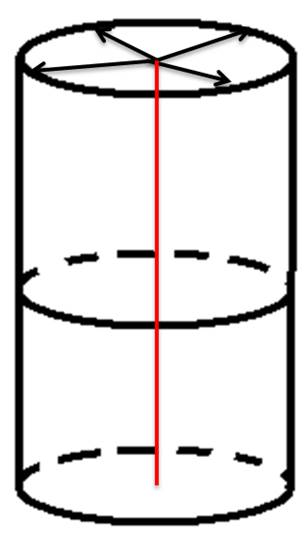}
  \caption{Two-dimensional hedgehog with director pointing radially outwards from axis of cylinder.}
  \label{2Dhedgehog}
\end{figure}
Since $|\nabla \tilde \n(\x)|^2=\frac{1}{r^2}$, it follows from \eqref{coercive} that  under the Ericksen inequalities \eqref{ericksenineq} 
$$I(\tilde\n)\geq c_0\int_\Omega |\nabla\n|^2d\x=2\pi c_0\int_0^Lr\cdot\frac{1}{r^2}\,dr=\infty,$$
so that $\tilde\n\not\in W^{1,2}(\om,S^2)$ and the {\it line defect} $\{(0,0,x_3):0<x_3<L\}$ has infinite energy according to the Oseen-Frank theory. 

Other more commonly observed line defects are the {\it index-$\half$} defects illustrated in Fig. \ref{indexhalf}, in which the corresponding line fields are parallel to the curves shown in the $(x_1,x_2)-$plane with zero $x_3$ component, and the line defects are in the $x_3-$direction at the points shown. (See \cite[Fig. 6]{ZhangKinlochWindle2006} for an interesting example of the defect in Fig. \ref{indexhalf} (a) occurring in a liquid crystalline phase of an aqueous suspension of carbon nanotubes.)  The terminology index-$\half$ means that the director rotates by half of $2\pi$ on a circle surrounding the defect. In particular the index-$\half$ defects are not orientable, as can be seen for example in the case Fig. \ref{indexhalf} (a) by trying to assign an orientation as in Fig. \ref{indexhalfa} (a). Thus by Theorem \ref{orientabilitythm} these line fields have infinite energy in the constrained Landau - de Gennes theory. 

Could we alter the line field just in a core encircling the line defect (see Fig. \ref{indexhalfa} (b)) so that the new line field has finite energy? For the two-dimensional hedgehog this is possible by `escape into the third dimension' (see, for example,  \cite[p 115ff]{stewart04}). However for the index-$\half$ defects it is not possible (while maintaining the uniaxiality constraint \eqref{unicon}) since the nonorientability argument works outside such a core, so that again we would have a contradiction to Theorem \ref{orientabilitythm}. In a sector such as shown in Fig. \ref{indexhalfa} (c) the line field is orientable, but the corresponding Oseen-Frank energy is still infinite (this follows from the preceding argument for at least one of the three sectors, and can be proved for each of them by applying  \cite[Theorem B.1]{BourgainBrezisMironescu2000} in a similar way as in Section \ref{orientability}).
\begin{figure}[tbp] 
  \centering
  \includegraphics[width=4in,height=2.29in,keepaspectratio]{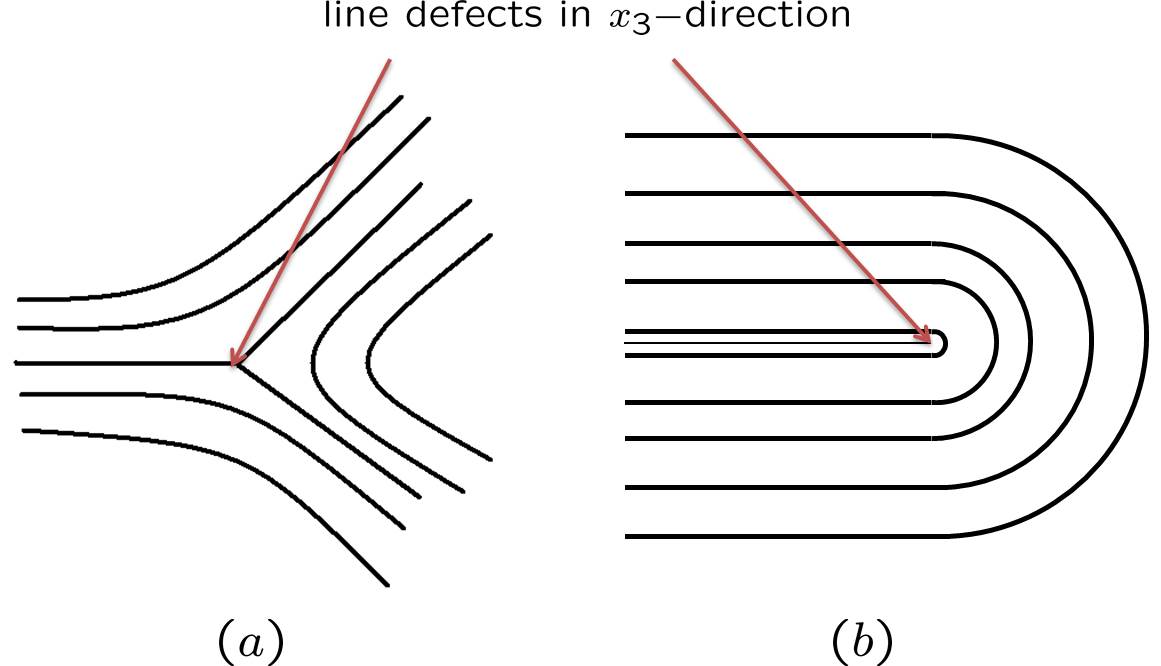}
  \caption{Two examples of index-$\half$ defects.}
  \label{indexhalf}
\end{figure}
\begin{figure}[tbp] 
  \centering
  \includegraphics[width=4.09in,height=1.62in,keepaspectratio]{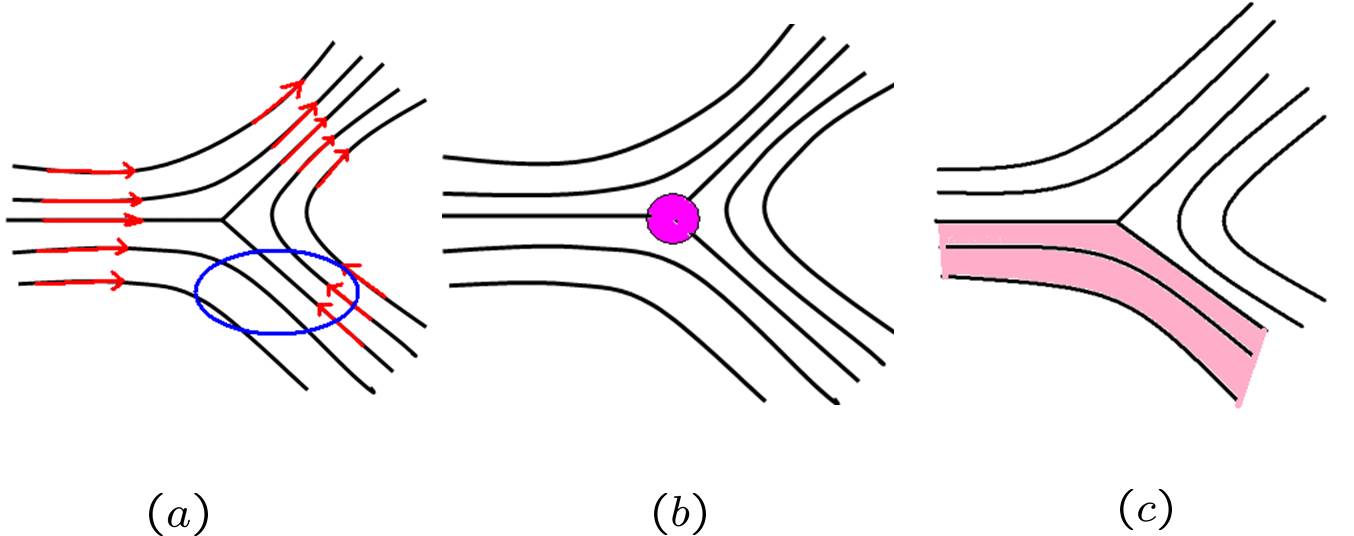}
  \caption{(a) argument showing that the index-$\half$ defect is not orientable, (b) the energy in the constrained Landau - de Gennes theory is still infinite however we alter  the line field in a core around the defect, (c) the Oseen-Frank energy is infinite in each sector.}
  \label{indexhalfa}
\end{figure}

That these line defects have infinite energy arises from the quadratic
growth in $\nabla\n$ of $W(\n,\nabla\n)$, which in turn follows from the quadratic growth of $\psi(\Q,\nabla\Q)$ in $\nabla\Q$ according to the derivation of the constrained theory. But  there is no reason to suppose
 that $W(\n,\nabla\n)$ is quadratic for large $|\nabla\n|$ (such as 
near defects). So a possible remedy would be to assume that $W(\n,\nabla\n)$
 has {\it subquadratic} growth, i.e. 
\be 
\label{subq}
W(\n,\nabla\n)\leq C(|\nabla \n|^p+1),
\ee 
where $1\leq p<2$, which would make line defects
 have finite energy. This can be done without affecting the behaviour of $W$ for small values of $\nabla\n$. For example, we can let
$$
W_\alpha(\n,\nabla\n)=\frac{2}{p\alpha}\left(\left(1+\alpha W(\n,\nabla \n)\right)^{\frac{p}{2}}-1\right),
$$
where   $\alpha>0$ is small. Then $W_\alpha(\n,\nabla\n)\to W(\n,\nabla\n)$ as
 $\alpha\to 0$.
 As shown in \cite{j67}, and assuming the Ericksen inequalities, $W_\alpha$ satisfies the growth 
conditions 
$$
C'_\alpha(|\nabla \n|^p-1)\leq W_\alpha(\n,\nabla\n)\leq C_\alpha|\nabla\n|^p,
$$
for positive constants $C_\alpha, C'_\alpha$. Setting 
$$I_\alpha(\n)=\int_\Omega W_\alpha(\n,\nabla\n)\,d\x,$$
 we obtain  that for the two-dimensional hedgehog $I_\alpha(\tilde\n)<\infty$ as desired.
 Also $W_\alpha(\n,\cdot)$ is convex.

Another undesirable consequence of the quadratic growth of $W(\n,\cdot)$ concerns the existence of finite energy configurations satisfying prescribed boundary conditions of physical interest. When $\om\subset \R^3$ has $C^2$ boundary and  
  ${\bf N}\in W^{\half,2}(\partial\om,S^2)$ , then
  (see Hardt \& Lin \cite[Theorem 6.2]{hardtlin})  there is an
$\n\in W^{1,2}(\om,S^2)$ with $\n|_{\partial\om}={\bf N}$. However the situation is different for less regular boundaries and boundary data.
Indeed, as shown by Bedford \cite{bedfordthesis} (an alternative proof can be based on \cite[Theorem B.1]{BourgainBrezisMironescu2000}), for the cube $Q=(-1,1)^3$ there is no $\n\in W^{1,2}(Q,S^2)$ satisfying the homeotropic boundary conditions $\n|_{\partial Q}=\nnu$, where $\nnu$ denotes the unit outward normal. However there are such $\n\in W^{1,p}(Q,S^2)$ for $1\leq p<2$, an example being given by\footnote{This can be verified by separately estimating $\nabla\n$ in neighbourhoods of the points where it is not smooth, namely $\x=0$, points on a cube edge, and corners of the cube.} 
\be 
\label{cubeex}
\n(\x)=\frac{\m(\x)}{|\m(\x)|}, \mbox{ where }\m(\x)=\left(\frac{x_1}{1-x_1^2}, \frac{x_2}{1-x_2^2},\frac{x_3}{1-x_3^2}\right),
\ee
so that for suitable $W$ with subquadratic growth in $\nabla\n$ there would be a corresponding energy minimizer having finite energy.

However, considering $W$ with subquadratic growth is insufficient by itself to handle the case of index-$\half$ defects due to their nonorientability. We return to this issue in the next section.

In the Ericksen theory (see \eqref{ericksenenergy}) we can model point and line defects 
by finite energy configurations in which $\n$ is
 discontinuous and $s=0$ at the defect (melting at the core). In this case there is
no need to change the growth rate at infinity. For example,
if we consider the special case when 
$$I_{LdG}(\Q)=\int_\Omega\left(\frac{K}{2}|\nabla \Q|^2+\psi_B( \Q)\right)\,d\x,$$
then the uniaxial ansatz
$$
\Q(\x)=s(\x)\left(\n(\x)\otimes \n(\x)-\frac{1}{3}{\bf 1}\right)
$$
gives the functional
$$I_E(s,\n)=\int_\Omega \left(\frac{K}{2}(|\nabla s|^2+2s^2|\nabla \n|^2)
 +\psi_B(s)\right)\,d\x,$$
where $\psi_b(s)=\hat\psi(\frac{2s^2}{3}, \frac{2s^3}{27})$. Then $\n$ can have a singularity at a point or curve and  
  $s$ can tend to zero sufficiently fast as the point or curve is
approached to make $I_E(s,\n)$ finite. However for non simply-connected domains or index-$\frac{1}{2}$ defects there is the same
 orientability
problem as in the Oseen-Frank theory.

\subsection{Planar defects}
\label{surfacedefects}
Following \cite{u14,j67}, and motivated by similar models from fracture mechanics (see \cite{francfortmarigo98,bourdinetal08}),  let us explore whether it might be reasonable to consider a free-energy functional for nematic and cholesteric liquid crystals of free-discontinuity type
\be 
\label{freedisc}
I(\n)=\int_\Omega W(\n,\nabla \n)\,d\x+\int_{S_\n}f(\n_+,\n_-,\nnu)\,d{\mathcal H}^2,
\ee
for $\n\in SBV(\Omega,S^2)$, where  $SBV(\Omega,S^2)$ denotes the space of special mappings of bounded variation taking values in $S^2$, $\nnu$ is the normal to the jump set $S_\n$ and $\n_+, \n_-$ the corresponding limits from either side of $S_\n$. The reader is referred to \cite{ambrosioetal00} for a comprehensive discussion of $SBV$, including an explanation of why $S_\n, \nnu,\n_+$ and $\n_-$ are well defined. Here $W(\n,\nabla \n)$ is assumed to have the Oseen-Frank form or be modified so as to have subquadratic growth as suggested in the previous section. 

We assume that the interfacial energy  $f:S^2\times S^2\times S^2\to [0,\infty)$ is continuous and frame-indifferent, i.e.
\begin{equation}
\label{surfacefi}f(\vR\n_+,\vR\n_-,\vR\nnu)=f(\n_+,\n_-,\nnu)
\end{equation}
for all $\vR\in SO(3), \n_+,\n_-,\nnu\in S^2$, and that $f$ is invariant to reversing the signs of $\n_+,\n_-$, reflecting the statistical head-to-tail symmetry of nematic and cholesteric molecules, so that 
\begin{equation}\label{surfacerev}
f(-\n_+,\n_-,\nu)=f(\n_+,-\n_-,\nnu)=f(\n_+,\n_-,\nnu).
\end{equation}
{A necessary and sufficient condition that $f$ satisfies $\eqref{surfacefi}, \eqref{surfacerev}$  is that (see \cite{u14,j67} and for a related result \cite{smith71})
\begin{eqnarray*}
f(\n_+,\n_-,\nnu)= 
 g((\n_+\cdot \n_-)^2, (\n_+\cdot\nnu)^2,(\n_-\cdot\nnu)^2,(\n_+\cdot \n_-)(\n_+\cdot\nnu)(\n_-\cdot\nnu))
\end{eqnarray*} 
for a continuous function $g:D\to[0,\infty)$, where
$$D=\{(\alpha,\beta,\gamma,\delta):\alpha,\beta,\gamma\in [0,1], \delta^2=\alpha\beta\gamma, \alpha+\beta+\gamma-2\delta\leq 1\}.$$
In the following subsections we consider various situations in which planar discontinuities of $\n$ and/or models such as \eqref{freedisc} are potentially of interest, referring the reader to \cite{j67} for more details.
\subsubsection{Nematic elastomers.}
Nematic elastomers are polymers to whose polymer chains rod-like mesogens are attached. Thus they combine features of nonlinear elasticity and liquid crystals. The energy functional for nematic elastomers
 proposed by Bladon, Terentjev \& Warner \cite{bladonetal} is given by
$$
I(\y,\n)=\int_\Omega \frac{\mu}{2}\left(D\y (D\y)^T\cdot L_{a, \n}^{-1}-3\right)\,d\x,
$$
where
$$
L_{a, \n}=a^\frac{2}{3}\n\otimes\n+a^{-\frac{1}{6}}({\bf 1}-\n\otimes\n)
$$
and $\mu>0, a>0$ are material parameters. Here $\y(\x)$ denotes the deformed position of the material point $\x\in\om$. As is usual for models of polymers the material is assumed incompressible, so that $\y$ satisfies the pointwise constraint $\det \nabla \y(\x)=1$ for $\x\in\om$. 

By minimizing 
 the integrand over $\n\in S^2$ we obtain the purely elastic energy
\be 
\label{one}
I(\y)=\int_\Omega W(\nabla\y)\,d\x, 
\ee 
where
$$
W(\A)=\frac{\mu}{2}\left(a^{-\frac{2}{3}}v_1^2(\A)+a^\frac{1}{3}(v_2^2(\A)+v_3^2(\A))\right),
$$
and $v_1(\A)\geq v_2(\A)\geq v_3(\A)>0$ denote the singular values of $\A$, 
that is the eigenvalues of $\sqrt{\A^T\A}$.

 The free-energy function  \eqref{one} is not quasiconvex \cite{desimonedolzmann},
 and admits minimizers in which $\nabla {\y}$ 
jumps across planar interfaces, so that the minimizing ${\bf n}$
 of the integrand also jumps. Stripe domains involving jumps in $\nabla\y$, similar to those seen in martensitic phase transformations (see, for example, \cite{bhattacharya03}), have been observed in  experiments of Kundler \& Finkelmann \cite{kundlerfinkelmann}. While the functional \eqref{one} ignores Frank elasticity,
 i.e. terms in 
$\nabla{\n}$, theories have been proposed in which such terms or corresponding terms in $\nabla\Q$ are included (see, for example, \cite{AndersonCarlsonFried1999,caldereretal2015}). The experimental observations suggest that it could be interesting to investigate whether in such models a corresponding  $SBV$  formulation  allowing jumps in ${\n}$ could be useful.

\subsubsection{Order reconstruction.}
We consider the situation illustrated in Fig. \ref{orderreconstruction}, in which a nematic liquid crystal occupies the region $\om_\delta=(0,l_1)\times(0,l_2)\times(0,\delta)$ of volume $|\om_\delta|=l_1l_2\delta$ between two parallel plates a small distance $\delta>0$ apart. The director $\n$ is subjected to antagonistic boundary conditions
$$\n(x_1,x_2,0)=\pm\e_1,\; \n(x_1,x_2,\delta)=\pm\e_3$$ on the plates, and periodic boundary conditions $$\n(0,x_2,x_3)=\n(l_1,x_2,x_3),\;\n(x_1,0,x_3)=\n(x_1,l_2,x_3)$$ on the other faces. 
\begin{figure}[h] 
  \centering
  \includegraphics[width=4.09in,height=1.55in,keepaspectratio]{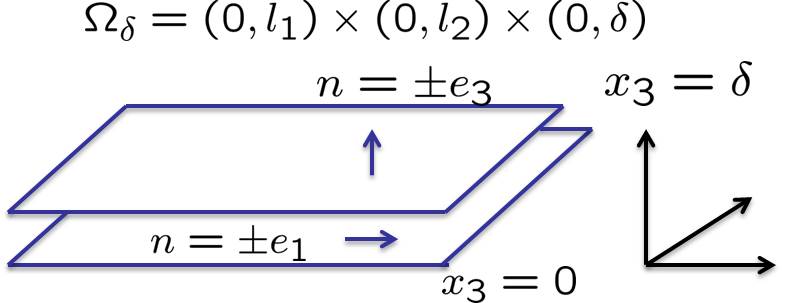}
  \caption{Thin plate with antagonistic boundary conditions.}
  \label{orderreconstruction}
\end{figure}
Similar problems have been considered by many authors using a variety of models (see, for example, \cite{ambrosiovirga91,barberietal04,barberobarberi,bisietal03,carboneetal,lamy14,palffyetal}). In \cite{j67} it is explained how using a Landau - de Gennes model, or molecular dynamics simulations \cite{zannonithinfilm}, leads for sufficiently small plate separation $\delta$ to a jump in the director (defined as in Section \ref{orderparameters} as the eigenvector of $\Q$ corresponding to it largest eigenvalue). Also in \cite{j67} it is shown that for a special choice of $W$ and $f$ in \eqref{freedisc} the minimum of $I$ is attained in $SBV(\om_\delta,S^2)$ satisfying the boundary conditions in a suitable sense. Here (see also \cite{j70}) we confine ourselves to showing  that in general, for $W(\n,\nabla\n)$ having the Oseen Frank form \eqref{OFe} with the Frank constants satisfying the Ericksen inequalities \eqref{ericksenineq}, for sufficiently small $\delta$ the infimum $I_{\rm inf}$ of $I(\n)$ among $\n\in SBV(\om_\delta,S^2)$ satisfying the boundary conditions is strictly less than the minimum of $I(\n)$ among $\n\in H^1(\om_\delta,S^2)$ satisfying the boundary conditions. Indeed, letting 
\be 
\label{test}
\N=\left\{\begin{array}{ll} \pm\e_1,& 0<x_3<\frac{\delta}{2}\\
\pm\e_3,&\frac{\delta}{2}<x_3\leq \delta \end{array}\right.,
\ee 
we deduce that $I_{\rm inf}\leq I(\N)=l_1l_2f(\e_1,\e_3,\e_3)$. On the other hand, if $\n\in W^{1,2}(\om_\delta,S^2)$ satisfies the boundary conditions, then
by \eqref{coercive}
\be 
  I(\n)&=&\int_{\om_\delta} W(\n,\nabla\n)\,d\x \nonumber\\
&\geq&c_0\int_{\om_\delta} |\nabla\n|^2d\x\nonumber\\
&\geq& c_0|\om_\delta|^{-1}\left|\int_{\om_\delta} \nabla\n\,d\x\right|^2\nonumber\\
&=&c_0|\om_\delta|^{-1}(l_1l_2)^2|\pm\e_3\mp\e_1|^2\nonumber\\
&=&2c_0 \frac{l_1l_2}{\delta},
\label{est}
\ee 
so that $I(\n)>I_{\rm inf}$ provided $\delta<\frac{2c_0}{f(\e_1,\e_3,\e_3)}$.
\subsubsection{Smectic thin films.}
A somewhat similar situation to the order reconstruction problem occurs in the experiments on smectic A thin films carried out by the research group of Emmanuelle Lacaze (see \cite{coursaultetal2014,coursaultetal2016,lacazemicheletal,michel04,michel06,zapponelacaze,zapponeetal,zapponeetal12}). Here there is parallel anchoring on the substrate, with homeotropic anchoring on the free surface of the film, leading to interesting configurations of the smectic layers in which their normals $\m$, and thus the director $\n$ also, suffer jump discontinuities on surfaces. The applicability of $SBV$ models for these experiments is currently being investigated.
\subsubsection{Recovering orientability}
Finally we indicate a purely mathematical application of $SBV$ models, to recover orientability of the director $\n$ in situations in which $\n$ is not orientable in smaller function spaces, by allowing $\n$ to jump to $-\n$ across suitable surfaces (see Fig. \ref{recover}). This can be formally associated with an energy functional of the form \eqref{freedisc} with a singular interfacial energy term
\be 
\label{singularf}
f(\n_+,\n_-,\nnu)=\left\{\begin{array}{ll}0 &\mbox{if }(\n_+\cdot\n_-)^2=1\\ \infty&\mbox{otherwise},\end{array}\right.
\ee 
which in turn can be considered as the limit $k\to\infty$ of $f(\n_+,\n_-,\nnu)=k(1-(\n_+\cdot\n_-)^2)$. That is, jumps from $\n$ to $-\n$ cost zero energy.
\begin{figure}[tbp] 
  \centering
  \includegraphics[width=4in,height=2.05in,keepaspectratio]{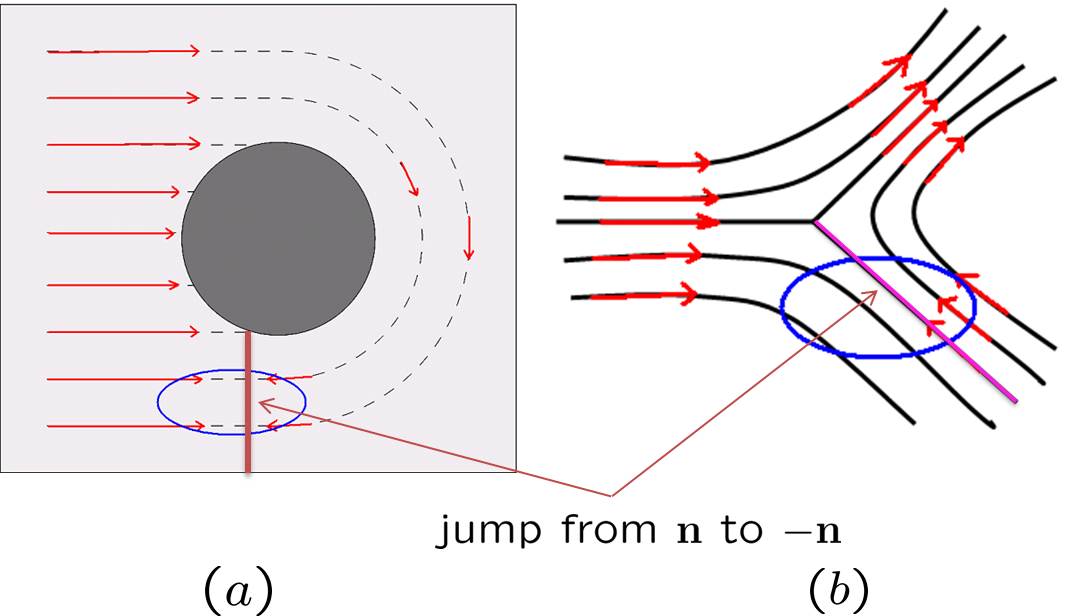}
  \caption{Recovering orientability in $SBV$ by allowing jumps of $\n$ to $-\n$ across suitable surfaces, $(a)$ for a smooth line field in a non simply-connected domain and $(b)$ for an index-$\half$ defect.}
  \label{recover}
\end{figure}
A corresponding lifting theorem is:
\begin{thm}[Bedford \cite{bedfordcholesterics}]
\label{SBVorient}
  Let $\om\subset\R^3$ be a bounded Lipschitz domain. $\Q=s\left(\n\otimes\n-\frac{1}{3}
{\bf 1}\right)\in W^{1,2}(\om,M^{3\times 3})$, where
$s\neq 0$ is constant. Then there exists a unit vector field ${\bf m}\in SBV(\om,S^2)$
such that $\m\otimes \m=\n\otimes \n$, and if $\x\in S_\m$ then $\m_+(\x)=-\m_-(\x)$.
\end{thm} 
Bedford \cite{bedfordcholesterics} also proves a related result in the context of the Ericksen theory. Theorem \ref{SBVorient} applies to the situation in Fig. \ref{recover} (a), but not to situations involving index-$\half$ singularities, for which an extension to $W^{1,p}$ would be required.

 \section*{Acknowledgements} This research  was supported by 
EPSRC
(GRlJ03466, the Science and Innovation award to the Oxford Centre for Nonlinear
PDE EP/E035027/1, and EP/J014494/1), the European Research Council under the European Union's Seventh Framework Programme
(FP7/2007-2013) / ERC grant agreement no 291053 and
 by a Royal Society Wolfson Research Merit Award. 
I  offer warm thanks to Elisabetta Rocca and Eduard Feireisl for organizing such an interesting programme, to the other lecturers and participants for the lively interaction, and to Elvira Mascolo and the CIME staff for the smooth and friendly organization in a beautiful location.

I am indebted to my collaborators Apala Majumdar, Arghir Zarnescu and Stephen Bedford for many discussions related to the material in these notes,   and to Apala Majumdar, Epifanio Virga, Claudio Zannoni and Arghir Zarnescu for  kindly reading the notes and pointing out various errors and infelicitudes.

\bibliography{gen2,balljourn,ballconfproc,ballprep}
\bibliographystyle{abbrv}

\printindex
\end{document}